\documentclass[11pt]{article}
\usepackage{color, colortbl}
\usepackage{graphicx,epstopdf}
\usepackage[margin = 1 in]{geometry}
\usepackage{enumerate}
\usepackage{amsmath,amssymb,amsthm}
\usepackage{algorithm}
\usepackage{algorithmic}
\usepackage{caption}
\usepackage{comment}
\usepackage{subcaption}
\usepackage{bm}
\usepackage{appendix}
\usepackage{multirow}
\usepackage{braket}
\usepackage[english]{babel}

\usepackage[usenames,dvipsnames]{xcolor}
\usepackage[hidelinks]{hyperref}
\hypersetup{colorlinks = true,
linkcolor = MidnightBlue, anchorcolor =red,
citecolor = ForestGreen,
filecolor = red,urlcolor = red,
            pdfauthor=author}

\usepackage[capitalize]{cleveref}
\usepackage[sort&compress,square,numbers]{natbib}
\usepackage{adjustbox}
\usepackage{xspace}
\usepackage{complexity}
\usepackage{soul}
\usepackage{tikz}
\usepackage{rotating}

\usepackage{authblk}

\numberwithin{equation}{section}

\newcommand{\diag}{\operatorname{diag}}

\newcommand{\tr}{\operatorname{Tr}}

\newcommand{\mc}[1]{\mathcal{#1}}
\newcommand{\wt}[1]{\widetilde{#1}}

\newcommand{\norm}[1]{\left\lVert#1\right\rVert}

\newcommand{\ketbra}[2]{|#1\rangle\!\langle #2 |}

\newcommand{\Or}{\mathcal{O}}

\newcommand{\CC}{\mathbb{C}}
\newcommand{\ZZ}{\mathbb{Z}}

\newtheorem{thm}{\protect\theoremname}
\theoremstyle{plain}
\newtheorem{lemma}[thm]{\protect\lemmaname}
\theoremstyle{plain}
\newtheorem{rem}[thm]{\protect\remarkname}
\theoremstyle{plain}
\theoremstyle{plain}
\newtheorem{prop}[thm]{\protect\propositionname}
\theoremstyle{plain}
\newtheorem{cor}[thm]{\protect\corollaryname}

\newtheorem{defn}[thm]{\protect\definitionname}

\newtheorem{fact}[thm]{Fact}

\providecommand{\definitionname}{Definition}
\providecommand{\assumptionname}{Assumption}
\providecommand{\corollaryname}{Corollary}
\providecommand{\lemmaname}{Lemma}
\providecommand{\propositionname}{Proposition}
\providecommand{\remarkname}{Remark}
\providecommand{\theoremname}{Theorem}

\newcommand{\REV}[1]{\textcolor{black}{ #1}}
\newcommand{\mmg}[1]{\textcolor{black}{ #1}}
\newcommand{\gap}{\mathrm{Gap}}

\newcommand{\ms}[1]{\mathscr{#1}}
\renewcommand{\R}{\mathbb{R}}
\renewcommand{\C}{\mathbb{C}}
\newcommand{\eps}{\epsilon}
\newcommand{\lad}{\lambda}
\newcommand{\si}{\sigma}

\def \ww {\omega}

\def \w {\widetilde}
\def \q {\quad}

\def  \mi {{\bf 1}}
\def \bh {\mc{B}(\mc{H})}
\def \l {\langle}
\def \r {\rangle}
\def \si {\sigma}
\def \dd {\cdot}
\def \ms {\mathsf}

\newcommand{\vp}{\varphi}

\def \xx {\mf{X}}
\def \yy {\mf{Y}}
\def \zz {\mf{Z}}
\newcommand{\mf}[1]{\mathsf{#1}}
\def \bx {\textbf{X}}

\def \bz {\textbf{Z}}

\title{Polynomial-Time Preparation of Low-Temperature Gibbs States for 2D Toric Code}
\author[1]{Zhiyan Ding\thanks{\texttt{zding.m@berkeley.edu}.}}
\author[2]{Zeph Landau\thanks{\texttt{zeph.landau@gmail.com}.}}
\author[3]{Bowen Li\thanks{\texttt{bowen.li@cityu.edu.hk}.}}
\author[1,4]{Lin Lin\thanks{\texttt{linlin@math.berkeley.edu}.}}
\author[5]{Ruizhe Zhang\thanks{\texttt{rzzhang@berkeley.edu}.}}
\affil[1]{\emph{Department of Mathematics, University of California, Berkeley}}
\affil[2]{\emph{Department of Computer Science, University of California, Berkeley}}
\affil[3]{\emph{Department of Mathematics, City University of Hong Kong}}
\affil[4]{\emph{Applied Mathematics and Computational Research Division, Lawrence Berkeley National Laboratory}}
\affil[5]{\emph{Simons Institute for the Theory of Computing}}
\date{}

\begin{document}

\maketitle

\begin{abstract}
We propose a polynomial-time algorithm for preparing the Gibbs state of the two-dimensional toric code Hamiltonian at any temperature, starting from any initial condition, significantly improving upon prior estimates that suggested exponential scaling with inverse temperature. Our approach combines the Lindblad dynamics using a local Davies generator with simple global jump operators to enable efficient transitions between logical sectors. Our proof also shows that the Lindblad dynamics with a digitally implemented low-temperature local Davies generator is able to efficiently drive the quantum state
towards the ground state manifold.
\end{abstract}

\section{Introduction}

The ability (or the lack of it) to efficiently prepare Gibbs states has far-reaching implications in quantum information theory, condensed matter physics, quantum chemistry, statistical mechanics, and optimization.
Given a quantum Hamiltonian $H\in\CC^{2^N\times 2^N}$,  we would like to prepare the associated thermal state $\sigma_\beta\propto e^{-\beta H}$, where  $\beta$ is the inverse temperature. A number of quantum algorithms have been designed to efficiently prepare high-temperature Gibbs states with a small $\beta=\Or(\mathrm{poly}(N^{-1})$)~\cite{PoulinWocjan2009,ChowdhurySomma2017,VanApeldoornGilyenGriblingEtAl2017,GilyenSuLowEtAl2019,an2023quantum}. However, as the temperature lowers (i.e., $\beta$ becomes large), the complexity of these algorithms can scale exponentially in the number of qubits $N$, rendering them impractical for low-temperature regimes where the Gibbs state has a significant overlap with the ground state of $H$.

Recent advancements have rekindled interest in designing quantum Gibbs samplers based on Lindblad dynamics~\cite{MozgunovLidar2020,RallWangWocjan2023,chen2021fast,ChenKastoryanoBrandaoEtAl2023, ChenKastoryanoGilyen2023,WocjanTemme2023,ding2024efficient}.
These algorithms rely on a specific form of open quantum system dynamics to drive the system toward its thermal equilibrium, an idea pioneered by Davies in the 1970s~\cite{davies1970quantum,Davies1974}.
The efficiency of such algorithms largely depends on the mixing time of the underlying dynamics, which can vary significantly across different systems and different forms of Lindbladians.

The computational complexity of preparing quantum Gibbs states, computing partition functions, and the potential for establishing quantum advantage in these tasks is a topic of ongoing debate in the literature.
On the one hand, at high enough temperatures, there exist polynomial-time classical algorithms to sample from Gibbs states and to estimate partition functions~\cite{bravyi2021complexity,MannHelmuth2021,YinLucas2023,BakshiLiuMoitraEtAl2024}. On the other hand, in the low-temperature regime, preparing classical Gibbs states is already \NP-hard in the worst case~\cite{Barahona1982,Sly2010}, and we do not expect efficient quantum algorithms in these cases.  The development of these new Gibbs samplers has also contributed to advancements in our complexity-theoretic understanding~\cite{rouz2024,bergamaschi2024quantum,rajakumar2024gibbs}. \cite{rouz2024} proved that simulating the Lindbladian  proposed in \cite{ChenKastoryanoGilyen2023} to time $T=\poly(N)$ at $\beta = \Omega(\log(N))$ for a $k$-local Hamiltonian is \BQP-complete. \cite{bergamaschi2024quantum,rajakumar2024gibbs}
constructed a family of $k$-local Hamiltonians such that quantum Gibbs sampling at constant temperatures (lower than the classically simulatable threshold)
can be efficiently achieved with the block-encoding framework~\cite{ChenKastoryanoBrandaoEtAl2023}, and the task is classically intractable, assuming no collapse of the polynomial hierarchy.

None of these constructions imply the efficient preparation of Gibbs states at low temperatures. Indeed, when the temperature is sufficiently low, the Gibbs state can exhibit a high overlap with the ground state, and cooling to these temperatures is expected to be \QMA-hard in the worst case.
However, it is important to recognize that \QMA-hardness does not preclude the possibility of developing efficient Gibbs samplers for \emph{specific} Hamiltonians.
Consequently, understanding and controlling the mixing time for specific systems (or specific classes of systems) is a fundamental open question in this field and can provide valuable insights into the practical performance of these Gibbs samplers~\cite{Alicki_2009,temme2015faststabilizer,Temme_2017,PhysRevB.90.134302,RevModPhys.88.045005,Freeman_2018,TemmeKastoryanoRuskaiEtAl2010,KastoryanoTemme2013,KastoryanoBrandao2016,BardetCapelGaoEtAl2023,kochanowski2024,Entropy_2024,rouz2024}.

In this work, we propose a novel Gibbs sampler with nonlocal jump operators,
and analyze its convergence rate for preparing \emph{low-temperature} Gibbs states of the 2D toric code~\cite{Kitaev2003}, a paradigmatic model in quantum information theory, quantum error correction, and condensed matter physics. In the context of thermalization, when the toric code is exposed to thermal noise modeled by a specific form of Lindbladians called the Davies generator, the seminal work of Alicki, Fannes, and Horodecki~\cite{Alicki_2009} showed that the inverse spectral gap (which leads to an upper bound of the mixing time) grows exponentially with inverse temperature $\beta$. Using a different argument based on energy barriers, Temme and Kastoryano~\cite{temme2015faststabilizer,Temme_2017} confirmed that the thermalization time (i.e., mixing time of the Lindblad dynamics) 
of the 2D toric code with local noise should indeed scale exponentially with $\beta$. However, it is unknown whether the $\exp(\beta)$ factor in the mixing time is unavoidable for \emph{all} Lindblad dynamics on this problem.
This leads to the central question of this work:
\begin{quote}\textit{Can we design a quantum algorithm, based on Lindblad dynamics, that efficiently prepares the low-temperature Gibbs state of the 2D toric code from \textbf{any} initial state with a polynomial runtime in both the inverse temperature $\beta$ and the number of qubits $N$?}
\end{quote}



\paragraph{Notations.} 
For a finite-dimensional Hilbert space $\mc{H}$, we denote by $\mc{B}(\mc{H})$ the space of bounded operators with identity element $\mi$. We let $\mc{D}(\mc{H}): = \{\rho \in \mc{B}(\mc{H})\,; \ \rho \ge 0\,,\  \tr(\rho) = 1\}$ be the set of quantum states. Denoting by $X^\dag$ the adjoint operator of $X \in \mc{B}(\mc{H})$, we recall the Hilbert--Schmidt (HS) inner product on $\bh$: $\l X, Y\r := \tr (X^\dag Y)$. With some abuse of notation, the adjoint of a superoperator $\Phi: \mc{B}(\mc{H}) \to \mc{B}(\mc{H})$ for $\l \dd,\dd \r$ is also denoted by $\Phi^\dag$. Moreover, $\{X, Y\} = XY + YX$ and
$[X,Y] = XY - YX$ for $X, Y \in \bh$ denote the anti-commutator and commutator, respectively.
We will use the standard asymptotic notations: $\Or$, $\Omega$, and $\Theta$. Precisely, we write $f=\Omega(g)$ if $g=\Or(f)$, and $f=\Theta(g)$ if $f=\Or(g)$ and $g=\Or(f)$. We shall use the standard Pauli matrices:
\[
\sigma^x = \begin{bmatrix}
0 & 1 \\
1 & 0
\end{bmatrix}, \quad
\sigma^y = \begin{bmatrix}
0 & -i \\
i & 0
\end{bmatrix}, \quad
\sigma^z = \begin{bmatrix}
1 & 0 \\
0 & -1
\end{bmatrix}.
\]

\subsection{Contribution} \label{sec:contri}

We derive a perhaps counterintuitive result regarding the 2D toric code: by employing a Lindbladian composed of a local Davies generator supplemented with simple global jump operators that facilitate transitions between logical sectors to overcome the energy barrier, it is possible to efficiently prepare the Gibbs state of the 2D toric code at any temperature. Furthermore, ignoring the information in the logical space, the local Davies generator alone suffices to efficiently drive the density matrix towards the ground state manifold
in the zero temperature limit.  This approach circumvents the previously established exponential dependence of the mixing time on $\beta$, achieving a polynomial scaling with system size (the number of qubits $N$).


To be specific, given $\beta>0$, we construct the following Davies generator:
\REV{\begin{equation}\label{eqn:fast_mixing_Davies}
\mathcal{L}_\beta = \underbrace{\sum^{N}_{j=1}\left(\mathcal{L}_{\sigma^x_j}+\mathcal{L}_{\sigma^y_j}+\mathcal{L}_{\sigma^z_{j}}\right)}_{:=\mathcal{L}_{\text{\rm local full}}}+\underbrace{\mathcal{L}_{\xx_1}+\mathcal{L}_{\zz_1}+\mathcal{L}_{\xx_2}+\mathcal{L}_{\zz_2}}_{:=\mathcal{L}_{\rm global}}\,,
\end{equation}}
where $\mc{L}_{[\cdot]}$ is the standard Davies generator defined via \cref{eqq:davies2} based on 2D toric code Hamiltonian \eqref{2d:toric_ham},
and $\xx_1,\zz_1,\xx_2,\zz_2$ are global logic operators for 2D toric code, see \cref{eqn:2D_global_X_Z} in \cref{sec:2D_toric}. Then, we prove the following main results:

\begin{thm}[Fast mixing of 2D toric code]\label{thm:fast_mixing_2D}
The spectral gap of the Gibbs sampler $\mathcal{L}_\beta$ in \eqref{eqn:fast_mixing_Davies} has the following lower bound
\begin{equation*}
  {\rm Gap}(- \mc{L}_\beta) = \max\left\{\exp(-\Or(\beta)),\Omega(N^{-2})\right\}\,.
\end{equation*}
\end{thm}

Our proof is valid for all temperatures, which includes two important regimes: 1. finite temperature regime $\beta = \mathcal{O}(1)$; 2. low temperature regime $\beta \gg 1$. In the first case of $\beta = \mathcal{O}(1)$, the local Davies generator $\mathcal{L}_{\text{\rm local full}}$ is sufficient to ensure the spectral gap of $\mc{L}_\beta$ at least $\exp(-\Theta(\beta))$ and independent of $N$. This is consistent with the result in \cite[Theorem 2]{Alicki_2009}. In the second case of $\beta \gg 1$, the Davies generator with global jumps $\mathcal{L}_{\rm global}$ significantly increases the spectral gap of $\mc{L}_\beta$, achieving a lower bound of order $\mathrm{poly}\left(N^{-1}\right)$ and independent of $\beta$.


The choice of the global jump operators in \eqref{eqn:fast_mixing_Davies} is natural. As $\beta\to \infty$, the thermal state converges toward the ground state of the Hamiltonian. Since the 2D toric code has four degenerate ground states that are not locally connected, the local Davies generator $\mathcal{L}_{\text{\rm local full}}$ cannot efficiently transit between these ground states, resulting in a slow mixing process. To overcome this difficulty, we introduce the global jump operators $\mathcal{L}_{\rm global}$ that enable transitions between different ground states at low temperatures, which ensures fast mixing even at zero temperature. This phenomenon is not unique to the 2D toric code. For example, in our paper, we also consider a simpler 1D Ising model in \cref{sec:1d_ising}, construct a similar Davies generator with global jump operators, and demonstrate a fast mixing result similar to \cref{thm:fast_mixing_2D} to illustrate the proof concept.

Our proof implies a more detailed characterization of the mixing process.
The algebra of observables for 2D toric code can be expressed as a tensor product between a logical observable space and a syndrome space. The spectral gap
of the Lindbladian in the logical space and the syndrome space can be analyzed independently. We find that for the standard Davies generator with only local jump operators, the exponentially vanishing spectral gap on $\beta$ is \emph{only} due to the action on the logical space, where the 2D toric code has four linearly independent ground states that are not locally connected. On the other hand,  we show that the spectral gap of the local Davies generator, when restricted to the syndrome space (by tracing out the logical subspace), has a lower bound that decays polynomially with the size of the system and remains independent of $\beta$ in the low-temperature regime. This is summarized in the following proposition:
\begin{prop}\label{prop:syndrome}
Let $\mathcal{L}_{\text{\rm local full}}$ be defined as in Eq.~\eqref{eqn:fast_mixing_Davies}. Then a part of $\mathcal{L}_{\text{\rm local full}}$ admits the syndrome subspace as an invariant subspace and exhibits a ``large'' spectral gap.

Specifically, there is a subset of Paulis $\{p_i\}\subset \{\sigma^{x/y/z}_j\}$, which defines $\mc{L}_{\rm local}=\sum_i\mc{L}_{p_i}$, a  decomposition $\mathcal{H}=\CC^{2^N}=\mc{H}_{\rm logic}\otimes \mc{H}_{\rm syndrome}$ with $\mc{H}_{\rm logic}\cong\mathbb{C}^4$ and $\mc{H}_{\rm syndrome}\cong\mathbb{C}^{2^{N-2}}$, and a corresponding decomposition $\mathcal{B}(\mc{H})\cong\mc{B}\left(\mc{H}_{\rm logic}\right)\otimes \mc{B}\left(\mc{H}_{\rm syndrome}\right)$, such that
\[
\mc{L}_{\rm local}\left(\mc{B}\left(\mc{H}_{\rm logic}\right)\otimes \mc{B}\left(\mc{H}_{\rm syndrome}\right)\right)=\mc{B}\left(\mc{H}_{\rm logic}\right)\otimes \mc{L}_{\rm local}\left(\mc{B}\left(\mc{H}_{\rm syndrome}\right)\right)\,.
\]
The spectral gap of $- \mc{L}_{\rm local}$ restricted to the syndrome space $\mc{B}\left(\mc{H}_{\rm syndrome}\right)$ is lower bounded by $$\max\left\{\exp(-\Or(\beta)), \Omega(N^{-2})\right\}\,.$$
\end{prop}
The above proposition implies that, even in the absence of the global jump operator, the system thermalizes quickly within the syndrome space. This is a key step towards proving \cref{thm:fast_mixing_2D}; see \cref{sec:technical_overview} for further details. Here the operator $\mathcal{L}_{\text{local}}$ in \cref{prop:syndrome} is carefully designed so that analyzing syndrome mixing reduces to studying the spectral gap of two distinct Glauber dynamics with quasi-1D classical Ising Hamiltonians (precisely, on the ``snake'' and ``comb'', see \cref{fig:2d_snake_comb}) at low temperature. Then, to derive lower bounds of their spectral gaps,  we introduce a new iterative method and establish a new tight estimate on the minimal eigenvalue of the perturbed graph Laplacian of the stair graph. To the best of our knowledge, this result concerning low-temperature thermalization for such quasi-1D classical Ising models is also novel in the literature.

\subsection{Implications}


\paragraph{Mixing time for preparing low-temperature Gibbs state.}
In this paper, we mainly focus on estimating the spectral gap of the Lindblad generator. It is well known that a lower bound on the spectral gap provides an upper bound on the mixing time of the Lindblad dynamics~\cite{TemmeKastoryanoRuskaiEtAl2010}. Specifically, for a detailed balanced Lindblad generator $\mathcal{L}$ with a spectral gap lower bounded by $\alpha$ and the unique fixed point $\si_\beta$, the following holds:
\[
\|e^{t \mathcal{L}^\dagger}\rho-\sigma_\beta \|^2_{\mathrm{tr}} \leq
\chi^2(e^{t \mathcal{L}^\dagger}\rho, \si_\beta) \le \chi^2(\rho, \si_\beta)
e^{- 2 \alpha t},
\]
where $\chi^2(\rho, \si_\beta)  = \mathrm{Tr}[(\rho-\sigma_\beta) \sigma_\beta^{-1/2}(\rho-\sigma_\beta) \sigma_\beta^{-1/2}]$ is the $\chi^2$-divergence. Notice that
\[
\max_\rho \chi^2(\rho, \si_\beta) \le (\lad_{\min}(\si_\beta))^{-1} \leq 2^{N}e^{\beta\|H\|}\,,
\]
with $\lad_{\min}(\dd)$ denoting the minimal eigenvalue. For 2D toric code Hamiltonian, we have the operator norm $\|H\| = \mc{O}(N)$. Thus, it readily gives an upper bound of the mixing time $t_{\mathrm{mix}}(\epsilon):=  \{t\ge 0\,;\  \|e^{t \mathcal{L^{\dag}} } \rho -\sigma_\beta \|_{\rm tr} \leq \epsilon, \  \forall \, \text{quantum states $\rho$}\}$ for a fixed $\epsilon$:
\[
t_{\rm mix}(\epsilon)=\mathcal{O}\left(\REV{\frac{N}{\alpha}\left(1 + \beta\right)}\right)\,.
\]
By substituting the spectral gap estimate from \cref{thm:fast_mixing_2D}, the mixing time of $\mc{L}_\beta$ in \cref{eqn:fast_mixing_Davies} scales as $\mathcal{O}\left(\min\left\{\beta\mathrm{poly}(N),\exp(c\beta)(N+\beta)\right\}\right)$ for some universal constant $c$, which is a significant improvement over the $\mathcal{O}\left(\exp(c\beta)(N+\beta)\right)$ bound given in~\cite{Alicki_2009} when $\beta\gg 1$. It is worth mentioning that in both cases, the mixing time scales polynomially with the number of qubits $N$, which is referred to as \emph{fast mixing} in the literature. In our case, the $N$ dependence arises from both the spectral gap and the prefactor $\max_\rho \chi^2(\rho, \si_\beta)$. While the latter dependence can sometimes be improved to $\mc{O}(\log(N))$ by considering the relative entropy distance~\cite{BardetCapelGaoEtAl2023,Entropy_2024,kochanowski2024}, achieving the \emph{rapid mixing} regime. This would also require translating the spectral gap lower bound into a constant or $\Omega(\log(N)^{-1})$ lower bound for the modified logarithmic Sobolev inequality (MLSI) constant. This presents an interesting direction for future exploration.


\paragraph{Thermal state versus ground state preparation.}  The 2D toric code is a stabilizer Hamiltonian, and its ground state can be efficiently prepared using local measurements. Recently, Hwang and Jiang~\cite{hwang2024gibbsstatepreparationcommuting} developed an efficient quantum protocol for preparing Gibbs states of the 2D toric code in  $\mathcal{O}(N^2)$ time uniformly in any positive temperature, and the approach can be generalized to the defected toric code. It constructs quantum Gibbs states by reducing the problem to Gibbs sampling of a classical local Hamiltonian. The goal of these protocols is to prepare ground or thermal states, starting from specific states such as tensor product states or maximally mixed states. In contrast, the goal of this work is different: it investigates whether dissipative dynamics can drive any initial state toward thermal equilibrium, analogous to certain dissipative state engineering approaches used for ground state preparation (see, e.g.,~\cite{VerstraeteWolfCirac2009}).


 However, these approaches do not easily generalize to algorithms for thermal state preparation. The fast mixing result at low temperatures in this work suggests that thermal state preparation methods can also serve as a generic tool for ground state preparation by selecting a sufficiently low temperature. Specifically, as discussed in \cref{prop:syndrome}, the local Davies generator within the syndrome space has a spectral gap that decays only polynomially with system size as \(\beta\) approaches infinity. This implies that at fixed sufficiently low temperature (\(\beta \gg 1\)), once a quasi-particle pair (an elementary excited state of the 2D toric code, see \cref{sec:2D_ground_state}) appears in the system, the local Davies generator can eliminate it in \(\text{poly}(N)\) time. Consequently, if the goal is to efficiently prepare \emph{some} ground state while discarding logical information, it is sufficient to use local Davies generators, and the time required is much shorter than that for equilibrating all logical sectors using local Davies generators.

\paragraph{Fast thermalization from a ground state.}  A natural question regarding the thermalization of the toric code Hamiltonian is: if one starts from the ground state, doesn't it take exponential time in $\beta$ to create even one quasi-particle pair in the first place?  While this is correct and may seem paradoxical, it does not contradict our main result that the thermalization time scales polynomially in $\beta$.
The reason is that although creating one quasi-particle pair from the ground state takes an exponential amount of time in $\beta$, the fraction of the excited state in the thermal state is exponentially small in $\beta$. In the case of 2D toric code,
the slow thermalization of local Davies generators at low temperatures mainly stems from the need to achieve an equal population across all ground states.  In other words, the challenge of thermalization lies mainly in transitioning between orthogonal ground states, as discussed in \cref{sec:contri}. In our work, we prove that the spectral gap of the global jump operators restricted to the logical space is independent of $\beta$ at low temperatures. This allows the system to equilibrate rapidly, even starting from a ground state.

\paragraph{Fast thermal state preparation,  quantum memory, and cellular automaton.}
A good \emph{self-correcting quantum memory} (SCQM) should be able to store a quantum state in contact with a cool thermal bath for a duration that increases exponentially with the size of the system.
This is sometimes also referred to as passively protected quantum memories, autonomous quantum error correction, or autonomous quantum memory protection~\cite{BarnesWarren2000,LeghtasKirchmairVlastakisEtAl2013,LieuLiuGorshkov2024}.
If the thermalization time of a quantum system only scales polynomially with system size, it cannot be considered a viable candidate for SCQM. \cite{Alicki_2009} demonstrated that at any constant temperature, the 2D toric code has a spectral gap that remains independent of system size and thus is not a good SCQM, and to date, valid candidates for SCQM are only known in 4D or higher~\cite{Alicki_2010}.
Our refined estimate indicates that local Davies generators at low temperatures (or even at zero temperature) can be efficient in annihilating quasi-particles in the syndrome space.

However, can we \emph{digitally} implement a low-temperature Davies generator to protect the logical information in the 2D toric code? Specifically, consider a Lindbladian $\mc{L}=\mc{L}_e+\mc{L}_r$ where $\mc{L}_e$ is a Davies generator modeling thermal noise in nature at some finite temperature $\beta^{-1}$, and $\mc{L}_r$ is a digitally implemented Davies generator at near-zero temperature, and the number of terms (each of up to unit strength) in $\mc{L}_r$ can scale polynomially in $N$. Starting from a pure ground state $\rho_0$ carrying well-defined logical information, we run the dynamics for some fixed time $t$, and then apply a single round of decoding map $\mc{E}_d$ to obtain a final density matrix $\rho_f$. Can we ensure that the trace distance between $\rho_0$ and $\rho_f$  decreases super-polynomially in $N$?
This procedure is similar to continuous-time recovery with time-independent control (also called continuous-time quantum error correction; see e.g.,~\cite{ChaseLandahlGeremia2008,KwonMukherjeeKim2022}). This task appears to be highly challenging.  The reason is that once $\mc{L}_e$ creates a syndrome in the form of a quasi-particle excitation, this quasi-particle can be diffused by either $\mc{L}_e$ or $\mc{L}_r$ at zero energy cost.
As a result, in the worst case, a quasi-particle may diffuse across the torus in polynomial time before it is annihilated, which changes the logical information.

The problem of dissipative error correction is closely related to the design of a local cellular automaton decoder for the 2D toric code, where each qubit updates its state based only on the states of its neighboring qubits and available syndrome information, following deterministic or probabilistic update rules. This process is inherently parallelizable.
Interestingly, a recent work by \cite{balasubramanian2024local} demonstrates that such a local cellular automaton decoder is indeed feasible, employing a hierarchical construction to address the challenge posed by diffusive quasi-particle excitations. This suggests that Lindblad dynamics incorporating a more complex global jump operator, designed to emulate this cellular automaton decoder, may also enable passive error correction for the 2D toric code.




\subsection{Related works}

The thermalization of stabilizer Hamiltonians using a local Davies generator has been explored in several prior works~\cite{Alicki_2009,temme2015faststabilizer,Temme_2017,PhysRevB.90.134302,RevModPhys.88.045005,Freeman_2018}. In \cite{Alicki_2009}, the authors demonstrated that the local Davies generator achieves fast thermalization for the 2D toric code. In particular, the spectral gap of the Davies generator, when considering all local Pauli coupling operators, is lower bounded by $\exp(-\Theta(\beta))$. The bound is valid
 for all temperatures and results in a mixing time (defined via the trace distance) scaling as $t_{\rm mix}=N\exp(\Theta(\beta))$. Along the same direction, \cite{temme2015faststabilizer,Temme_2017} considered general stabilizer codes and established a lower bound on the spectral gap  using the generalized energy barrier $\overline{\epsilon}$~\cite[Definition 13]{Temme_2017}. Specifically, for any given $\beta > 0$, the spectral gap can be lower bounded by $C_N\exp\left(-\beta \overline{\epsilon}\right)$, where $C_N$ is a technical constant that typically scales as $1/N$.
Although these works address thermalization across all temperatures, these lower bounds on the spectral gap are insufficient for efficiently preparing low-temperature thermal states, as the gap decays exponentially in \(\beta\). Specifically, when using a local Davies generator to transition between ground states along a local Pauli path, the energy must first increase, requiring the dynamics to overcome the energy barrier to fully thermalize. Furthermore, as analyzed in~\cite{Temme_2017, kastoryano2024}, an exponentially small spectral gap of order $\exp(-\Theta(\beta))$ in a local Davies generator appears inevitable when such energy barriers are present.

In our work, to overcome the bottleneck posed by the energy barrier,
we modify the local Davies generator by incorporating appropriate global coupling operators that can directly connect the degenerate ground states, thereby avoiding the issues associated with the generalized energy barrier defined by local Pauli paths. By refining the analysis in~\cite{Alicki_2009}, we demonstrate that the resulting dynamics exhibit a spectral gap that decays polynomially with the size of the system but remains \emph{independent of} $\beta$, ensuring fast mixing even at low temperatures. We note that the polynomially decaying spectral gap primarily arises from the mixing rate within the syndrome space, which remains unaffected by the introduction of the new global jump operators acting on the logical space. Consequently, our refined mixing time estimate rigorously demonstrates that the local Davies generator at low temperatures (and even at zero temperature) can efficiently annihilate quasi-particles residing in the syndrome space. A similar phenomenon has been numerically observed in the study of quantum memory~\cite{PhysRevB.90.134302,RevModPhys.88.045005,Freeman_2018}.

There is extensive literature on the mixing properties of local Davies generators for general local commuting Hamiltonians~\cite{KastoryanoBrandao2016,BardetCapelGaoEtAl2023,kochanowski2024,Entropy_2024}. However, most studies consider a fixed finite temperature $\beta = \mathcal{O}(1)$~\cite{KastoryanoBrandao2016,BardetCapelGaoEtAl2023,Entropy_2024,kochanowski2024}, particularly in the context of 1D local commuting Hamiltonians, where the mixing time implicitly depends on $\beta$, or a high temperature $\beta \ll 1$~\cite{KastoryanoBrandao2016,kochanowski2024}. Extending these general approaches to low-temperature thermal state preparation and explicitly calculating the temperature dependence remains an interesting and challenging problem.


There is also a long line of works studying the fast mixing and rapid mixing of different types of classical Markov chains for spin systems. However, many classical results suffer from \emph{low-temperature bottlenecks}: the dynamics can mix in polynomial time only when the inverse temperature $\beta$ is below some threshold $\beta_c$. For example, it is well-known that the Glauber dynamics for the ferromagnetic Ising model with $N$ spins on a $d$-dimensional lattice has mixing time $\Theta(N\log N)$ when $\beta < \beta_c(d)$ \cite{mo94,ls16}, and $e^{\Theta(N^{1-1/d})}$ when $\beta>\beta_c(d)$ \cite{pis96}, where $N$ is the total number of spins and the constant in the $\Theta$ notation depends on $\beta$. Some classical results managed to overcome this bottleneck by carefully designing the initial distribution \cite{gs22}, or studying other Markov chains or other graphical models \cite{js93,gj17,gsv19,hpr19,bcp21,cgg21}. However, we note that these classical techniques cannot be applied directly to obtain our results.

\subsection{Organization}

In the following part of the paper, we start with a technical overview of the proof of \cref{thm:fast_mixing_2D} in \cref{sec:technical_overview}. Then, we provide a brief introduction to the Davies generator and properties of its spectral gap in \cref{sec:preliminary}. The detailed introduction to the 2D toric code and its proof can be found in \cref{sec:2D_toric}. Additionally, in \cref{sec:1d_ising}, we discuss a simpler case of the 1D ferromagnetic Ising chain for completeness.


\subsection*{Acknowledgment}

This work is supported by the National Key R$\&$D Program of China Grant No. 2024YFA1016000
(B.L.), the U.S. Department of Energy, Office of Science, National Quantum Information Science Research Centers, Quantum Systems Accelerator (Z.D.),  by the Challenge Institute for Quantum Computation (CIQC) funded by National Science Foundation (NSF) through grant number OMA-2016245 (L.L.), and
by DOE Grant No. DE-SC0024124 (R.Z.).
L.L. is a Simons Investigator in Mathematics.
We thank Garnet Chan, Li Gao, Jiaqing Jiang, Yunchao Liu, and John Preskill for helpful discussions and feedback.

\section{Technical overview}\label{sec:technical_overview}

In this section, we provide a technical overview of the proof of \cref{thm:fast_mixing_2D}. In the analysis, there are three main steps:
\begin{enumerate}[(a)]
     \item \label{a} Decompose $\mathcal{H}$ into logic and syndrome subspaces and the observable algebra $\mathcal{B}(\mc{H})$ correspondingly.
  \item \label{b} Demonstrate efficient transition between logic subspaces.
  \item \label{c} Demonstrate fast mixing inside the syndrome subspace.
\end{enumerate}

The decomposition in the first step leverages the special structure of the stabilizer Hamiltonian, following the approach outlined in previous work by~\cite{Alicki_2009}. Specifically, for the 2D toric code, we decompose $\mc{H} = \mc{H}_{\rm logic} \otimes \mc{H}_{\rm syndrome}$ according to the eigendecomposition of $H^{\rm toric}$. Since $H^{\rm toric}$ has a four-dimensional ground state space, it encodes two logical qubits, making $\dim(\mc{H}_{\rm logic}) = 4$. The syndrome subspace $\mc{H}_{\rm syndrome}$ is then spanned by the electric and magnetic excited states, which are characterized by the bond configurations of the local observables in $H^{\rm toric}$.
That is, it can be further decomposed into the electric and magnetic excited subspaces $\mc{H}_{\rm syndrome}=\mc{H}^{\rm m}_{\rm b}\otimes \mc{H}^{\rm e}_{\rm b}$. According to the decomposition of $\mc{H}$, we can naturally decompose the observable algebra $\mc{B}(\mc{H}) = \mc{Q}_1 \otimes \mc{Q}_2 \otimes \mc{A}^{\rm full}_{\rm m} \otimes \mc{A}^{\rm full}_{\rm e}$, where $\mc{Q}_1 \otimes \mc{Q}_2$ is generated by the logical operators, i.e., the global operators $\xx_1, \zz_1, \xx_2, \zz_2$ appearing in \eqref{eqn:fast_mixing_Davies}. The syndrome algebras $\mc{A}^{\rm full}_{\rm m}$ and $\mc{A}^{\rm full}_{\rm e}$ are spanned by linear transformations acting on the syndrome subspaces $\mc{H}^{\rm m}_{\rm b}$ and $\mc{H}^{\rm e}_{\rm b}$, respectively. Here, the $\otimes$ symbol represents multiplication between commuting matrices, and every element in $\mc{Q}_1, \mc{Q}_2, \mc{A}^{\rm full}_{\rm m}$, and $\mc{A}^{\rm full}_{\rm e}$ is understood as a matrix defined over the entire Hilbert space $\mc{H}$. We put the detailed discussion of the above decomposition in \cref{sec:2D_ground_state}.

After decomposing $\mc{B}(\mc{H})$, we can further decompose the Davies generator \eqref{eqn:fast_mixing_Davies}:
\[\mc{L}_\beta = \underbrace{\mc{L}_{\rm local}+\mc{L}_{\rm global}}_{:=\mc{L}^{\rm gapped}}+\mc{L}^{\rm rest}\,,\]
with $\mc{L}^{\rm gapped}$ defined in \eqref{eqn:L_gapped} and $\mc{L}^{\rm rest}=\mc{L}_{\text{\rm local full}}-\mc{L}_{\rm local}$. Here we can show that $\mc{L}^{\rm gapped}$ is ergodic, i.e.,  $\mathrm{Ker}(\mc{L}^{\rm gapped})={\rm Span}\{\mi\}$.  By \cref{lem1} (item 2), we can lower bound $\mathrm{Gap}(\mc{L}_\beta)$ by $\mathrm{Gap}(\mc{L}^{\rm gapped})$. More importantly, the generator $\mc{L}^{\rm gapped}$ is block diagonal with respect to the following decomposition:
\begin{equation*}
\mc{B}(\mc{H})=\bigoplus_{B_1\in \{\mi,\xx_1, \yy_1,
\zz_1\},B_2\in \{\mi,\xx_2,\yy_2,\zz_2\}}\mc{B}_{B_1,B_2}\,,
\end{equation*}
where $\mc{B}_{B_1,B_2}=B_1\otimes B_2\otimes \mathcal{A}^{\rm full}_{\rm m}\otimes \mathcal{A}^{\rm full}_{\rm e}$ and the operators $\xx_i, \yy_i, \zz_i$ are given in \eqref{eqn:2D_global_X_Z}. Because $\mathrm{Ker}(\mc{L}^{\rm gapped})= {\rm Span}\{\mi\}$ and $\mc{L}^{\rm gapped}$ is block diagonal, we have
\begin{multline} \label{eqn:Gap_inequality}
\mathrm{Gap}\left(-\mc{L}_\beta\right)\geq \mathrm{Gap}\left(-\mc{L}^{\rm gapped}\right) \\ \geq \min\left\{\mathrm{Gap}\left(-\mc{L}^{\rm gapped}\middle|_{\mc{B}_{\mi,\mi}}\right),\lambda_{\min}\left(-\mc{L}^{\rm gapped}|_{\mc{B}_{B_1,B_2}}\right)\middle|B_1\neq \mi\ \text{or}\ B_2\neq \mi\right\}\,.
\end{multline}
Thus, for a lower bound estimate of $\mathrm{Gap}(-\mc{L}^{\rm gapped})$, it suffices to consider  $\mathrm{Gap}(-\mc{L}^{\rm gapped}|_{\mc{B}_{\mi,\mi}})$ and $\lad_{\min}(-\mc{L}^{\rm gapped})|_{\mc{B}_{B_1,B_2}}$ with $B_1\neq \mi\ \text{or}\ B_2\neq \mi$. The latter term characterizes the transition rate between different logical subspaces. If $\mc{L}_\beta$ contains only local coupling operators, the system requires a long evolution time to overcome the energy barrier and to transit between different logical subspaces, which leads to a slow mixing time scaling as $\exp(\Theta(\beta))$ according to~\cite{Alicki_2009}. In our work, an important observation is that the global logical operators $\xx_1, \zz_1, \xx_2, \zz_2$ can directly flip the logical qubits, enabling transitions between different logical subspaces without the need to overcome the energy barrier. The efficient transition in logic density space implies a fast decaying of $\exp(t \mc{L}^{\rm gapped})$ in the logic subspace $\{\mc{B}_{B_1,B_2}\}_{B_1\neq \mi\ \text{or}\ B_2\neq \mi}$ and a lower bound of $\{-\mc{L}^{\rm gapped}|_{\mc{B}_{B_1,B_2}}\}_{B_1\neq \mi\ \text{or}\ B_2\neq \mi}$. Specifically, in \cref{prop:second_step} in \cref{sec:gibbs_2d}, we show
\begin{equation}\label{eqn:Gap_inequality_2}
  -\mc{L}^{\rm gapped}|_{\mc{B}_{B_1,B_2}}\succeq \mathrm{Gap}\left(-\mc{L}_{\rm local}\middle|_{\mc{B}_{\mi,\mi}}\right)\,.
\end{equation}

According to the above analysis, the remaining thing to lower bound the spectral gap of $\mc{L}_\beta$ is to study the spectral gap of $\mc{L}_{\rm local}$ on $\mc{B}_{\mi,\mi}$. Roughly speaking, this requires us to prove \cref{prop:syndrome}. A similar task is done in~\cite{Alicki_2009}. However, we emphasize that the lower bound and proof technique in~\cite{Alicki_2009} is \textbf{not} suitable for our purpose. In~\cite[Proposition 2]{Alicki_2009}, while the spectral gap is independent of the system size, it decays as $\exp(-\Theta(\beta))$, indicating slow mixing in the syndrome subspace at low temperatures. A main contribution of this work is to
 show that this lower bound is not tight when $\beta \gg 1$, and the system actually mixes fast in the syndrome subspace at low temperatures. To this end, we develop a new iteration argument, a decomposition trick, and a tight estimation of the minimal eigenvalue of a perturbed graph Laplacian
 to prove that $\mathrm{Gap}\left(\mc{L}_{\rm local}\middle|_{\mc{B}_{\mi,\mi}}\right)\geq \min\left\{\exp(-\Theta(\beta)),\mathrm{poly}(1/N)\right\}$, which provides a much sharper lower bound of the spectral gap in the lower temperature regime. This result is summarized in \cref{prop:gap_quasi_1D} in \cref{sec:gap_quasi_1D}, which also provides a proof of \cref{prop:syndrome}. Plugging this into \eqref{eqn:Gap_inequality} and \eqref{eqn:Gap_inequality_2}, we can conclude \cref{thm:fast_mixing_2D}.








\section{Preliminaries}\label{sec:preliminary}



Let $H$ be a quantum many-body Hamiltonian on $\mc{H} \cong \C^{2^N}$ and $\si_\beta = e^{- \beta H}/\mc{Z}_\beta$ be the associated thermal state, where $\beta$ is the inverse temperature and $\mc{Z}_\beta = \tr(e^{- \beta H})$ is the partition function. In this section, we shall recall the canonical form of the Davies generator with $\si_\beta$ being the invariant state and some basic facts for its spectral gap analysis.

A superoperator $\Phi: \bh \to \bh$ is a quantum channel if it is completely positive and trace preserving (CPTP).
Lindblad dynamics \REV{$\mathcal{P}^\dagger_t$} is a $C_0$-semigroup of quantum channels with the generator defined by $\mc{L}^\dag(\rho): = \lim_{t \to 0^+} t^{-1}(\mc{P}^\dag_t (\rho) - \rho)$ for $\rho \in \mc{D}(\mc{H})$. Here and in what follows, we adopt the convention that the adjoint operators (i.e., those with $\dag$) are the maps in Schr\"{o}dinger picture acting on quantum states. Both generators $\mc{L}$ and $\mc{L}^\dag$ are usually referred to as Lindbladian. Davies generator is a special class of Lindbladians derived from the weak coupling limit of open quantum dynamics with a large thermal bath.

We first introduce the Bohr frequencies of $H$ by
\begin{equation*}
    B_H := \{\ww = \lad_i - \lad_j\,:\ \forall~\lad_i,\lad_j \in {\rm Spec}(H)\}\,,
\end{equation*}
where ${\rm Spec}(H)$ is the spectral set of $H$. Let $\{S_a\}_{a \in \mc{A}}$ be a set of coupling operators with $\mc{A}$ being a finite index set that satisfies
\begin{equation}\label{eq:Sa_self_adjoint}
    \{S_a\}_{a \in \mc{A}} = \{S_a^\dag\}_{a \in \mc{A}}\,.
\end{equation}
The jump operators $\{S_a(\ww)\}_{a,\ww}$  for the Davies generator are given by the Fourier components of the  Heisenberg evolution of $S_a$:
\begin{equation} \label{eq:jumpfou}
    e^{i H t} S_a e^{- i H t} = \sum_{\ww \in B_H} e^{i \ww t} \sum_{\lambda_i-\lambda_j=\omega} P_iS_aP_j:=\sum_{\ww \in B_H} e^{i \ww t} S_a(\ww)\,.
\end{equation}
where $P_{i/j}$ is the projection into the eigenspace $\lambda_{i/j}$.
By~\eqref{eq:Sa_self_adjoint}, we have $S_a(\omega)^\dagger = S_a(-\omega)$ for any $\omega\in B_H$.

We then introduce the Davies generator in the Heisenberg picture:
\begin{equation} \label{eqq:davies}
    \mc{L}_{\beta}(X) := \sum_{a\in \mc{A}} \mc{L}_{S_a}(X)\,,\q X \in \mc{B}(\mc{H})\,,
\end{equation}
with
\begin{equation} \label{eqq:davies2}
     \mc{L}_{S_a}(X): = \sum_{\ww \in B_H} \gamma_a(\ww) \left(S_a(\ww)^\dag\,X\,  S_a(\ww) - \frac{1}{2}\left\{S_a(\ww)^\dag  S_a(\ww), X\right\}  \right)\,,
\end{equation}
where the transition rate function $\gamma_a(\ww) > 0$ is given by the Fourier transform of the bath autocorrelation function satisfying the KMS condition \cite{kossakowski1977quantum}:
\begin{equation*}
    \gamma_a(-\ww) = e^{\beta \ww} \gamma_a(\ww)\,.
\end{equation*}
In this work, we always choose the transition rate function $\gamma_a(\ww)$ as the Glauber form:
\begin{equation}\label{eqn:Glauber}
    \gamma_a(\ww) = \frac{2}{e^{\beta \ww} + 1}\,,
\end{equation}
For later use, we define $g_a(\ww) := e^{\beta \ww/2} \gamma_a(\ww)$ and find $g_a(\ww) = g_a(-\ww)$. Then, letting
\begin{equation} \label{eq:la}
L_a(\ww) = \sqrt{g_a(\ww)} S_a(\ww)\,,
\end{equation}
we reformulate the Davies generator \eqref{eqq:davies}--\eqref{eqq:davies2} as follows:

\begin{equation} \label{eq:davies2}
    \begin{aligned}
        \mc{L}_\beta(X)
        & = \sum_{a \in \mc{A}} \sum_{\ww \in B_H} e^{-\beta \ww/2} \Big(L_{a}(\ww)^\dag X L_{a}(\ww) - \frac{1}{2}\left\{L_{a}(\ww)^\dag L_{a}(\ww), X\right\} \Big) \\
        & = \frac{1}{2}\sum_{a \in \mc{A}} \sum_{\ww \in B_H} e^{-\beta \ww/2} L_{a}(\ww)^\dag [X, L_{a}(\ww)] + e^{\beta \ww/2} [L_{a}(\ww), X] L_{a}(\ww)^\dag \\
        & = \frac{1}{2}\sum_{a \in \mc{A}} \sum_{\ww \in B_H} \gamma_a(\ww) S_{a}(\ww)^\dag [X, S_{a}(\ww)] + \gamma_a(-\ww) [S_{a}(\ww), X] S_{a}(\ww)^\dag\,.
    \end{aligned}
\end{equation}
where the second step follows from
\begin{align*}
    L_a(\omega)^\dagger=e^{\beta\omega/4}\sqrt{\gamma_a(\omega)}S_a(-\omega) = e^{\beta\omega/4}\sqrt{e^{-\beta \omega}\gamma_a(-\omega)}S_a(-\omega) = L_a(-\omega)\,.
\end{align*}

We now define the GNS inner product associated with the Gibbs state $\si_\beta$:
\begin{equation}\label{eqn:GNS_inner_product}
    \l Y, X \r_{\si_\beta} = \tr(Y^\dag X \si_\beta)\,.
\end{equation}
It is known \cite{kossakowski1977quantum,ding2024efficient} that the Davies generator satisfies the GNS detailed balance:
\begin{equation*}
     \l Y, \mc{L}_\beta(X) \r_{\si_\beta} =  \l \mc{L}_\beta(Y), X\r_{\si_\beta}\,,
\end{equation*}
and thus the associated Lindblad dynamics $e^{t \mc{L}_\beta^\dag}$ admits $\si_\beta$ as an invariant state, i.e.,
\begin{equation*}
 \mc{L}_\beta^\dag (\si_\beta) = 0\,.
\end{equation*}
It follows that $\mc{L}_\beta$ is similar to a self-adjoint operator for the HS inner product, called the \emph{master Hamiltonian}, and has only the real spectrum. To be precise, we introduce the transform $\vp_X:= X {\si_\beta}^{1/2}$, which gives $\l Y, X \r_{\si_\beta} = \l \vp_Y, \vp_X \r$. Then, the master Hamiltonian $\w{\mc{L}}_\beta$ is given by the similar transform of $\mc{L}_\beta$ via $\vp_X$:
\begin{equation} \label{eq:master_hami}
    \w{\mc{L}}_\beta := \vp \circ \mc{L}_\beta \circ \vp^{-1}\,,
\end{equation}
satisfying $\l Y, \mc{L}_\beta (X)\r_{\si_\beta} = \l \vp_Y, \w{\mc{L}}_\beta \vp_X\r$. It is easy to see that $- \w{\mc{L}}_\beta$ is positive semi-definite and $\sqrt{\si_\beta}$ is the zero-energy ground state of $-\w{\mc{L}}_\beta$. Thus, the spectral gap of $\mc{L}_\beta$ is the same as the ground state spectral
gap of the Hamiltonian $- \w{\mc{L}}_\beta$.

We say that $\mc{L}_\beta$ is primitive if $\si_\beta$ is the unique invariant state; equivalently, the kernel $\ker(\mc{L}_\beta)$ is of one dimension, spanned by $\mi$. In this case, we have
\begin{align*}
    \lim_{t \to \infty} e^{t \mc{L}_\beta^\dag}(\rho) = \si_\beta\,,\q \forall \rho \in \mc{D}(\mc{H})\,,
\end{align*}
and the spectral gap ${\rm Gap}(\mc{L}_\beta)$ of the primitive Davies generator can be characterized by the variational form:
\begin{equation*}
    {\rm Gap}(\mc{L}_\beta) = \inf_{X \neq 0\,, \tr(\si_\beta X) = 0} \frac{\l X, -\mc{L}_\beta (X) \r_{\si_\beta}}{\l X, X\r_{\si_\beta}}\,.
\end{equation*}
Moreover, thanks to the detailed balance, the operator norm of $\mc{L}_\beta$ can be computed by
\begin{equation} \label{eq:supernorm}
    \norm{\mc{L}_\beta}_{\si_\beta \to \si_\beta} = \sup_{X \neq 0} \frac{\l X, - \mc{L}_\beta (X) \r_{\si_\beta}}{\l X, X\r_{\si_\beta}}\,.
\end{equation}

By \cite{wolf5quantum,zhang2023criteria}, a sufficient and necessary condition for the primitivity of $\mc{L}_\beta$ is the $\C$-algebra generated by all the jump operators $\{S_a(\ww)\}_{a,w}$ is the whole algebra $\bh$. By this condition, one can check that for the choice of $\{S_a\}_{a \in \mc{A}} = \{\si_i^x, \si_i^y, \si_i^z\}_{i = 1}^N$, the associated Davies generator is always primitive. Without loss of generality, when discussing the Gibbs samplers in \cref{sec:1d_ising} and \cref{sec:2D_toric}, we always let $\{\si_i^x, \si_i^y, \si_i^z\}_{i = 1}^N$ be a subset of $\{S_a\}_{a \in \mc{A}}$ to guarantee the primitivity.

The following lemmas are collected from \cite[Lemmas 1 and 2]{Alicki_2009} with proof omitted, which shall be used repeatedly in our subsequent spectral gap analysis.

\begin{lemma} \label{lem1}

Let $A, B \in \bh$ be positive operators on a Hilbert space $\mc{H}$, i.e., $A, B \ge 0$. The spectral gaps of $A$ and $B$ are
their smallest non-zero positive eigenvalues, denoted by ${\rm Gap} (A)$ and ${\rm Gap} (B)$, respectively. We have:

\begin{itemize}
    \item If $A$ has a non-trivial kernel and $A^2 \ge g A$ for some real $g > 0$, then
    \begin{equation} \label{eq1}
        {\rm Gap} (A) \ge g\,.
    \end{equation}
    \item If $\ker(A + B)$ is non-trivial such that $\ker(A + B) = \ker(B)$, then
    \begin{equation} \label{eq2}
        \gap(A + B) \ge \gap(B)\,.
    \end{equation}
    \item If $A$ and $B$ are commuting and $\ker(A + B)$ is non-trivial, then
    \begin{equation} \label{eq3}
        \gap(A + B) \ge \min\{\gap(A),\gap(B)\}\,.
    \end{equation}
    \item  If $A$ has gap lower bound $g_A$ and $\l \vp, B \vp \r \ge g_B$ for all normalized $\vp \in \ker(A)$, then
    \begin{equation} \label{eq4}
        A + B \succeq \frac{g_A g_B}{g_A + \norm{B}}\,,
    \end{equation}
    where $\norm{B}$ denotes the operator norm of $B$.
\end{itemize}
\end{lemma}

\begin{rem}
    The second statement in \cref{lem1} means that for any primitive Davies generator with $\si_\beta$ being invariant, adding any other Davies generator keeping the invariant state does not decrease the spectral gap.
\end{rem}

\section{2D toric code}\label{sec:2D_toric}

Let $N = 2 L^2$ spins be on the edges of the toroidal lattice, modeled by the Hilbert space $\mc{H} \cong \C^{2^N}$. The Hamiltonian is given by
\begin{equation} \label{2d:toric_ham}
    H^{\rm toric}=-\sum_{s}\textbf{X}_s-\sum_{p}\textbf{Z}_p\,,
\end{equation}
where indices $s$ and $p$ denote a \emph{star} and \emph{plaquette} that consist of four sites around a node of the lattice and the center of a cell, respectively, as drawn in \cref{fig:2d_toric}, and the associated observables $\textbf{X}_s$ and $\textbf{Z}_p$ are given by
\begin{equation*}
    \textbf{X}_s = \prod_{i\in s}\sigma^x_{i}\,,\q \textbf{Z}_p = \prod_{i\in p}\sigma^z_{i}\,,
\end{equation*}
which commute with each other: $[\textbf{X}_s,\textbf{Z}_p] = [\textbf{X}_s,\textbf{X}_{s'}] = [\textbf{Z}_p,\textbf{Z}_{p'}]\equiv 0$ for any stars $s,s'$ and plaquettes $p, p'$. In addition, due to the periodic boundary condition, it holds that
\begin{equation}\label{eqn:periodic_boundary}
\prod_s \textbf{X}_s=1\,,\quad \prod_p\textbf{Z}_p=1\,,
\end{equation}



This section is devoted to the spectral gap analysis of the Davies generator (Gibbs sampler) in \eqref{eqn:fast_mixing_Davies} and the proof of \cref{thm:fast_mixing_2D}, building on the discussion in \cref{sec:technical_overview}. We will first discuss the ground states and the decomposition of the observable algebra of $H^{\rm toric}$ in \cref{sec:2D_ground_state}. Then, in \cref{sec:gibbs_2d}, we first address the step \eqref{b} outlined in \cref{sec:technical_overview}, reducing the problem to step \eqref{c}: the study of the spectral gap of the local Davies generator on the syndrome space. This will be explored in detail in \cref{sec:gap_quasi_1D}.

\begin{figure}[bthp]
\centering
\includegraphics[width=0.5\textwidth]{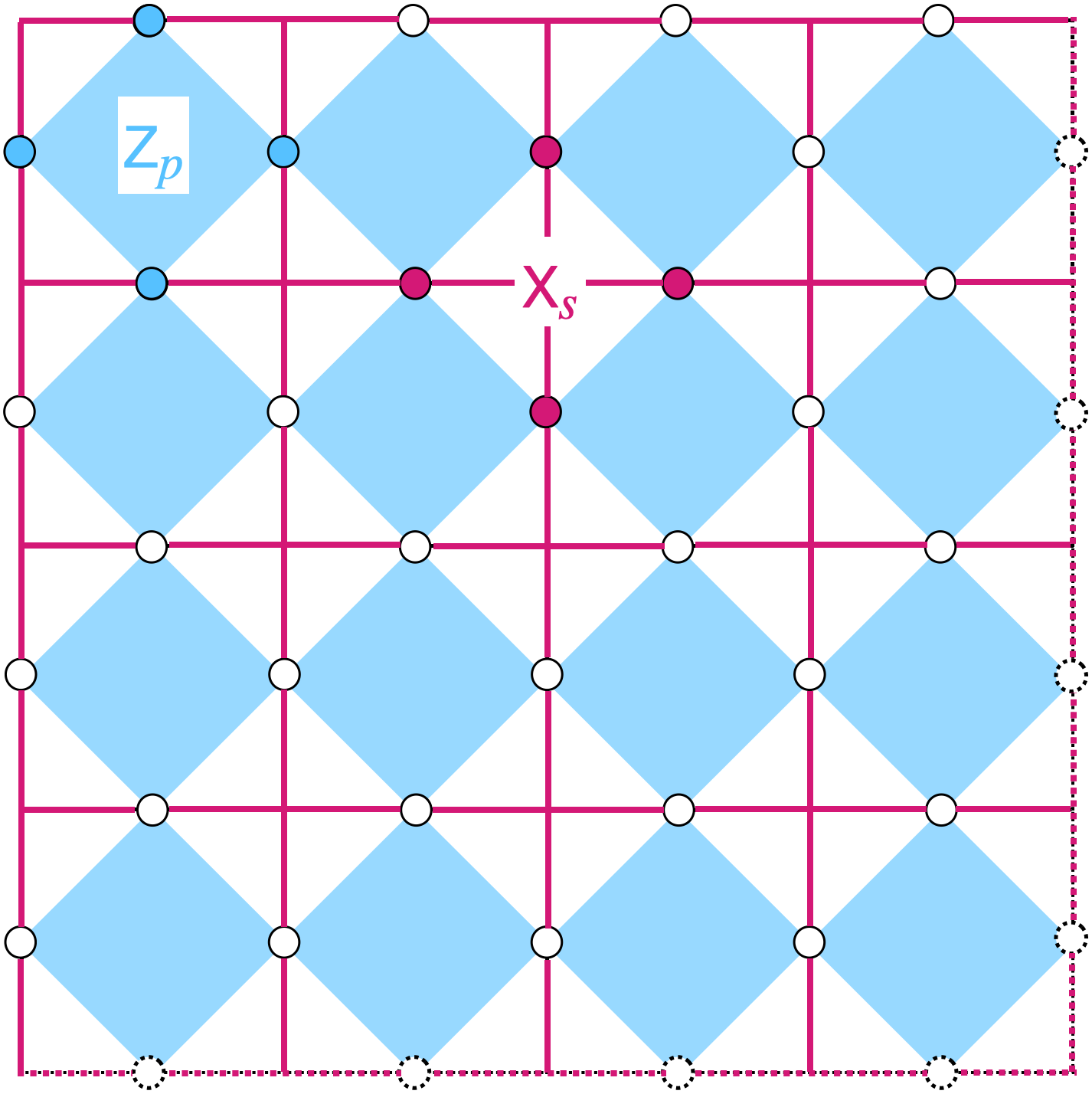}
\caption{2D toric code. Each blue plaquette contains one 4-local operator $\textbf{Z}_p=\prod_{i\,\in\,\text{plaquette }p}\sigma^z_{i}$. Each red star contains one 4-local operator $\textbf{X}_s=\prod_{i\,\in\,\text{star }s}\sigma^x_{i}$. The dashed line on the right/bottom is the same as the solid line on the left/top to indicate the periodic boundary condition.}
\label{fig:2d_toric}
\end{figure}

\subsection{Ground states and observable algebra}\label{sec:2D_ground_state}


In this section, we introduce the ground state space of $H^{\rm toric}$ and the associated observable algebra $\mc{B}(\mc{H})$, following \cite{Alicki_2009,Alicki_2010,bombin2013introduction}, which is important for the step \eqref{a} of the road map in \cref{sec:technical_overview}.

Analogous to the 1D ferromagnetic Ising chain (see \cref{sec:1d_ising}), the 2D  toric code Hamiltonian $H^{\text{toric}}$ is frustration-free, i.e., its ground state is simultaneously an eigenvector with eigenvalue 1 for all local terms $\mathbf{X}_s$ and $\mathbf{Z}_p$. Due to the toric structure with periodic boundary conditions, there are two topologically protected degrees of freedom (i.e., two logical qubits), resulting in a four-dimensional ground state space.




We now construct the ground state space of $H^{\rm toric}$ explicitly. Given a vector $\ket{\phi}$ satisfying $\bz_p \ket{\phi} = \ket{\phi}$ for all $p$, e.g., $\ket{0
^N}$ or $\ket{1^N}$ (note that there are many others),  we can construct a ground state as follows:
\begin{equation}\label{eqn:construct_ground_state}
    \ket{\psi} = \bx_{\rm star} \ket{\phi}\,,\q \bx_{\rm star}: = \prod_{s\in \rm star}\left(I+\textbf{X}_s\right) =  \sum_{\alpha \in \{0,1\}^{N}} \prod_{s\in \rm star} \bx_s^{\alpha_s}\,,
\end{equation}
which, as one can easily check, satisfies $\bz_p \ket{\psi} =  \ket{\psi}$, $\bx_s  \ket{\psi} =  \ket{\psi}$ for all $p,s$. To find all the ground states, we first define some global observables for two logical qubits (see \cref{fig:2d_global_X_Z}):
\begin{equation}\label{eqn:2D_global_X_Z}
\begin{aligned}
    & \xx_1 := \prod_{j \in \text{orange dots}} \sigma^x_j\,, \q \zz_1 := \prod_{j \in \text{green squares}} \sigma^z_j\,, \q \yy_1 := \REV{- i}\zz_1\xx_1\,, \\
    & \xx_2 := \prod_{j \in \text{blue dots}} \sigma^x_j\,, \q \zz_2 := \prod_{j \in \text{red squares}} \sigma^z_j\,, \q \yy_2 := \REV{- i}\zz_2\xx_2\,.
\end{aligned}
\end{equation}
Here, $\{\xx_1,\yy_1,\zz_1\}$ and $\{\xx_2,\yy_2,\zz_2\}$ generates two commuting Pauli algebras. Moreover, the operators $\xx_1$ and $\xx_2$ give one-to-one correspondences between four topological equivalent classes \cite[Fig.\,3.2]{Alicki_2010} (see also \cite[Section 3.2]{bombin2013introduction}). Specifically, we consider four orthogonal vectors that satisfy $\bz_p \ket{\phi} = \ket{\phi}$:
\begin{equation*}
    \ket{\phi_o} = \ket{0^N}\,,\q \ket{\phi_|} = \xx_1\ket{0^N}\,,\q  \ket{\phi_-} = \xx_2  \ket{0^N}\,,\q  \ket{\phi_+} = \xx_1 \xx_2 \ket{0^N}\,.
\end{equation*}
They generate the four orthogonal ground states via \eqref{eqn:construct_ground_state}:
\begin{equation}\label{eqn:2D_basis}
    \begin{aligned}
        &\ket{\psi_g} = \bx_{\rm star} \ket{\phi_g} \quad \text{with} \quad g = o,\ |,\ -,\ +\,,
    \end{aligned}
\end{equation}
which span the whole ground state space of $H^{\rm toric}$. In this form, $\ket{\psi_g}$ with $g = o,\ |,\ -,\ +$ corresponds to the logical qubits $\ket{00}$, $\ket{10}$, $\ket{01}$, and $\ket{11}$, respectively.



\begin{figure}[bthp]
    \centering
    \includegraphics[width=0.5\textwidth]{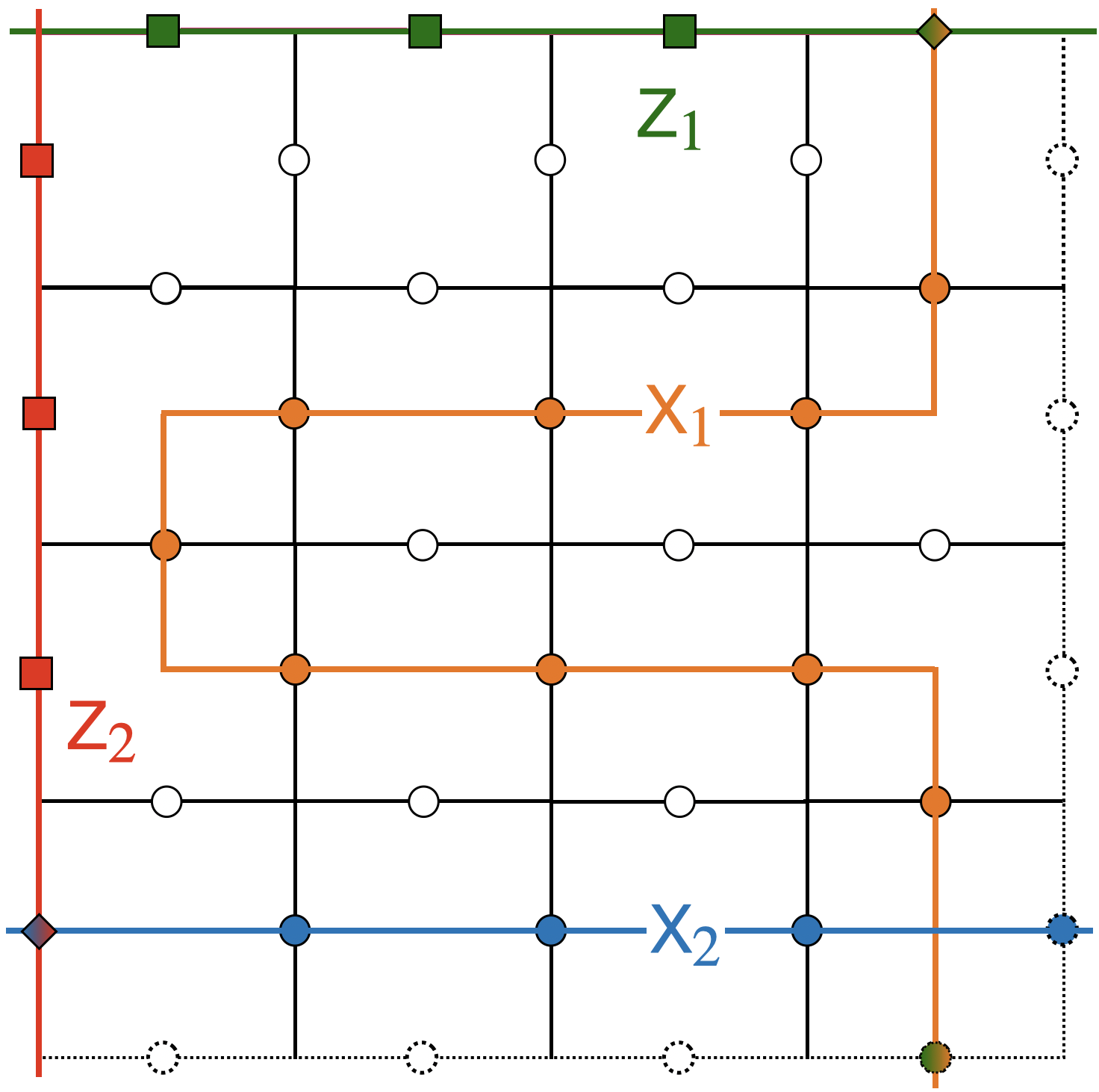}
    \caption{Four global logic operators in 2D toric code. Orange: $\xx_1 = \prod_{j \in \text{orange dots}} \sigma^x_j$. Green: $\zz_1 = \prod_{j \in \text{green squares}} \sigma^z_j$. Blue: $\xx_2 = \prod_{j \in \text{blue dots}} \sigma^x_j$. Red: $\zz_2 = \prod_{j \in \text{red squares}} \sigma^z_j$.
    Squares represent the qubits on $\zz_1, \zz_2$ logic operators, while dots represent the qubits on $\xx_1, \xx_2$ logic operators.}
    \label{fig:2d_global_X_Z}
\end{figure}

Next, to formulate the observable algebra, we first consider how the bit flip will influence the energy, namely, the excitation of $H^{\rm toric}$. For this, we delicately construct a snake path (as a subset of $N = 2L^2$ sites) passing through the centers of all the cells and a comb path passing through all the nodes of the toroidal lattice (see \cref{fig:2d_snake_comb}) such that
\begin{itemize}
    \item $\{\si_j^x\}_{j \in {\rm snake}}$ generates all the $X$-type  excitations (also called ``magnetic'' excitations):
    for any two plaquettes ($p_1, p_2$), we can find a path $l$ on the snake connecting them (blue path in \cref{fig:2d_snake_comb}).
    Then we define
\begin{equation}\label{eqn:magnetic_operator}
    W_l^m = \prod_{j \in l} \sigma^x_j\,,
\end{equation}
and then have the excited state $\ket{\psi_{p_1, p_2}}=W_l^m\ket{\psi}$ satisfying
\begin{equation*}
   \REV{\bz_{p_1} \ket{\psi_{p_1,p_2}} = \bz_{p_2} \ket{\psi_{p_1,p_2}}=- \ket{\psi_{p_1,p_2}},\quad H^{\rm toric} \ket{\psi_{p_1,p_2}} = (E_{\rm ground} + 2)\ket{\psi_{p_1,p_2}}\,,}
\end{equation*}
where $\ket{\psi}$ is a ground state of $H^{\rm toric}$.
\item $\{\si_j^z\}_{j \in {\rm comb}}$ generates all the $Z$-type excitations (also called ``electric''  excitations): for any two stars ($s_1, s_2$), we can find a path $l$ on the comb  connecting them (red path in \cref{fig:2d_snake_comb}). We then define
    \begin{equation}\label{eqn:electric_operator}
        W_l^e = \prod_{j \in l} \sigma^z_j\,.
    \end{equation}
    For a ground state $\ket{\psi}$, we define the excited state $\ket{\psi_{s_1,s_2}} =  W_l^e \ket{\psi}$ satisfying 
    \REV{\begin{equation*}
        \bx_{s_1} \ket{\psi_{s_1,s_2}} =
        \bx_{s_2} \ket{\psi_{s_1,s_2}}=
        - \ket{\psi_{s_1,s_2}},\quad H^{\rm toric} \ket{\psi_{s_1,s_2}} = (E_{\rm ground} + 2)\ket{\psi_{s_1,s_2}}\,.
    \end{equation*}}
    thanks to \REV{$\bx_{s_1(\text{or}\ s_2)} W_l^e = - W_l^e \bx_{s_1(\text{or}\ s_2)}$}.
    \item The snake and comb form a partition of all but two spins.
    In addition, the snake does not intersect with $\zz_1,\zz_2$, and the comb does not intersect with $\xx_1,\xx_2$. This ensures that the excitation operators $W_l^{m/e}$ constructed above commute with global observables: for a ground state $\ket{\psi}$,
    \begin{equation*}
        \mf{O} W_l^{m/e} \ket{\psi} = W_l^{m/e}  \mf{O} \ket{\psi}\,,
    \end{equation*}
    where $\mf{O} = \xx_1$, $\xx_1$, $\zz_1$, and $\zz_2$.
\end{itemize}
These elementary excited states $\ket{\psi_{p_1, p_2}}, \ket{\psi_{s_1,s_2}}$ are often called quasi-particle pairs (or quasi-particles for short). All excited states of $H^{\rm toric}$ can be expressed using these quasi-particles.




Following the above discussions, we decompose the space $\mathcal{H}$ according to the magnetic excitations observed by $\bz_p$ and the electric excitations observed by $\bx_s$:
\begin{equation} \label{eq:space_decom}
    \mc{H} = \C^2 \otimes \C^2 \otimes \mc{H}_{\rm b} \cong \C^{2^{N = 2L^2}} \q \text{with}\q  \mc{H}_{\rm b} = \mc{H}_{\rm b}^{\rm m}\otimes \mc{H}_{\rm b}^{\rm e}\,,
\end{equation}
where $\mc{H}_{\rm b}^{\rm m/e}$ is the space of electric/magnetic excited states spanned by
\begin{equation} \label{eq:excitedstate}
\Big\{\prod_l W_{l}^{\rm m/e}\ket{\psi_g}\,:\  W_{l}^{\rm m/e}\ \text{given in \eqref{eqn:magnetic_operator} and \eqref{eqn:electric_operator}}, g=o,|,-,+ \Big\}\,.
\end{equation}
Moreover, a basis vector in $\mc{H}_{\rm b}$ can be identified as
\begin{equation} \label{eq:basis2d}
    \ket{m}\ket{e} = \ket{m_1,\ldots m_{L^2}}\ket{e_1,\ldots e_{L^2}}\in \left\{+1,-1\right\}^{2L^2}\,,
\end{equation}
such that $\#\left\{m_j=-1\right\}, \#\left\{e_j=-1\right\} \in 2\ZZ$, due to the periodic boundary condition, with $m_j = \pm 1$ (resp., $e_j = \pm 1$) denoting the observation under $\bz_p$ (resp., $\bx_s$). Here and in what follows, we sort $\{\bz_p\}_p$ and $\{\bx_s\}_s$ along the snake and comb such that they can be indexed by $j \in \{1,\cdots, L^2\}$ (see \cref{fig:2d_snake_comb}). We emphasize that $\mathrm{dim}(\mc{H}_b)=2^{2L^2-2}$ follows from the parity constraint \eqref{eqn:periodic_boundary}, while each $\ket{m}\ket{e}$ is a $2L^2$-dimensional vector.

Now, we are ready to decompose the observable algebra $\mc{B}(\mc{H})$. Let $\mc{Q}_1$ and $\mc{Q}_2$ be the observable algebras over two logical qubits, generated by $\xx_1$ and $\zz_1$ and by $\xx_2$ and $\zz_2$, respectively. We denote by $\mc{A}_{\rm m/e}^{\rm full}$ the linear operator spaces on $\mc{H}_{\rm b}^{\rm m/e}$. Then we have the following decomposition of the observable algebra for the 2D toric code. The proof is straightforward and given in \cref{subsec:proof} for completeness.

\begin{lemma}\label{lem:2D_decomposition}
    The algebra of observables for 2D toric code can be decomposed into
  \begin{equation} \label{eq:opdecom}
    \mc{B}(\mc{H}) \cong \mc{Q}_1 \otimes \mc{Q}_2 \otimes \mc{A}_{\rm m}^{\rm full} \otimes \mc{A}_{\rm e}^{\rm full}\,,
\end{equation}
associated with the decomposition \eqref{eq:space_decom}, where the algebras $\mc{A}^{\rm full}_{\rm m}$ and $\mc{A}^{\rm full}_{\rm e}$ are generated by $\{\mathbf{Z}_p\}_p\cup \{\sigma^x_j\}_{j\in \text{\rm snake}}$ and $\{\mathbf{X}_s\}_s\cup\{\sigma^z_j\}_{j\in \text{\rm comb}}$. In addition, the four subalgebras $\mc{Q}_1$, $\mc{Q}_2$, $\mc{A}_{\rm m}^{\rm full}$, and $\mc{A}_{\rm e}^{\rm full}$ are commutative with each other, namely, the operators in these subalgebras acting on the entire Hilbert space $\mathcal{H}$ mutually commute with each other.
\end{lemma}

To give further interpretation on the decomposition \eqref{eq:opdecom}, we define four subspaces $\mc{H}_g$ with $g = o,\ |,\ -,\ +$ spanned by the ground state $\ket{\psi_g}$ in \eqref{eqn:2D_basis} and its excited states $W_l^{m/e} \ket{\psi_g}$ via \eqref{eqn:magnetic_operator} and \eqref{eqn:electric_operator}. Then it is easy to see $\dim(\mc{H}_g) = 2^{N-2}$ with $N = 2 L^2$. The algebra $\mathcal{Q}_1\otimes \mathcal{Q}_2$ gives all the linear transformation between
 subspaces $\mc{H}_g$ with $g = o,\ |,\ -,\ +$, while on each subspace $\mc{H}_g$, the linear maps are characterized by $\mc{A}^{\rm full}_{\rm m}\otimes \mc{A}^{\rm full}_{\rm e}$. In particular, recalling the basis \eqref{eq:basis2d}, $\mc{A}^{\rm full}_{\rm m}$ (resp., $\mc{A}^{\rm full}_{\rm e}$) acts only on  $\ket{m}$ (resp., $\ket{e}$), for example,
 \begin{equation*}
     \mc{A}^{\rm full}_{\rm m} (\ket{m}\ket{e}) = \left(\mc{A}^{\rm full}_{\rm m} \ket{m}\right)\ket{e}\,.
 \end{equation*}

\begin{figure}[bthp]
    \centering
    \includegraphics[width=0.5\textwidth]{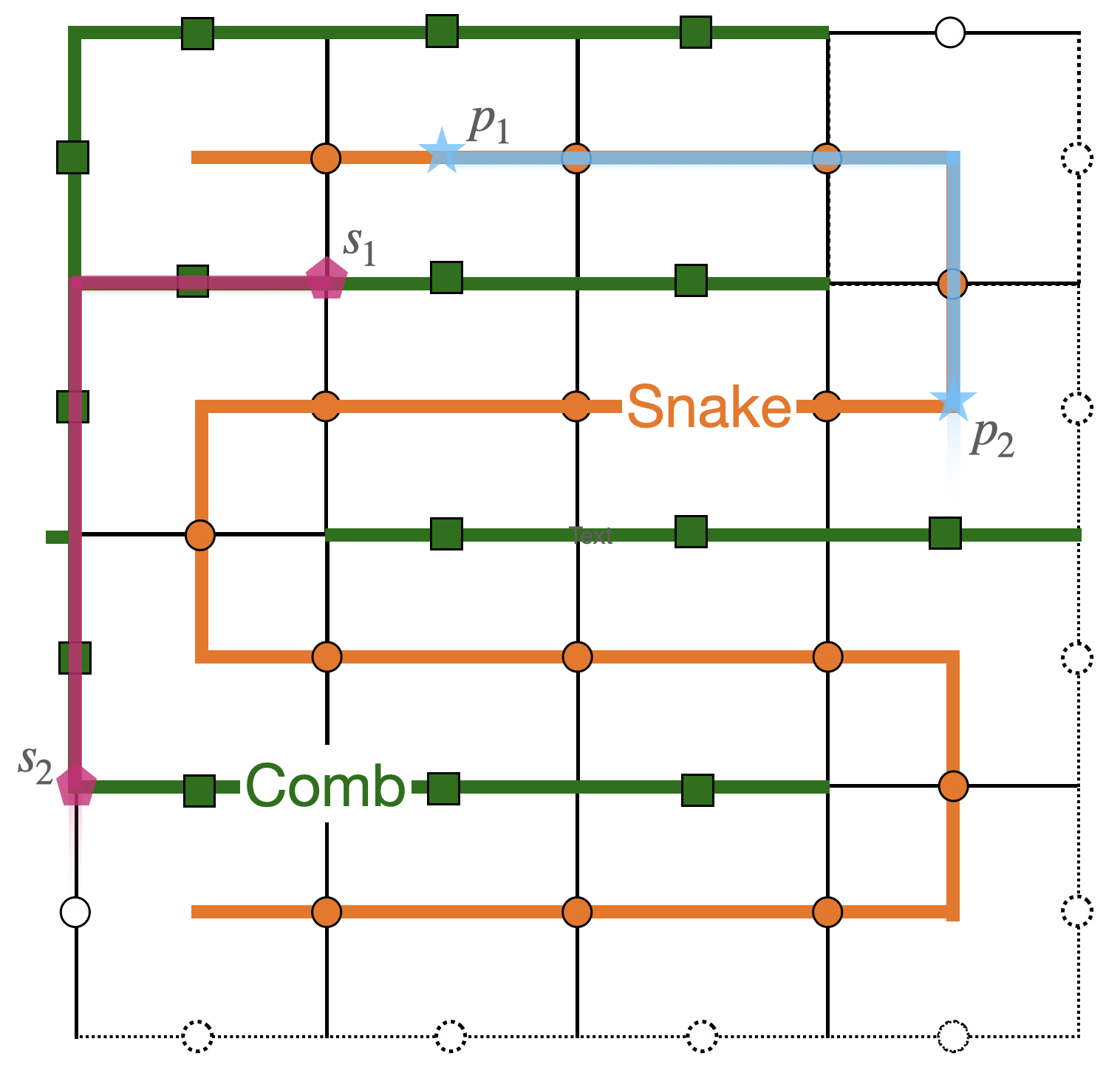}
    \caption{Orange dot Snake and green square comb for 2D toric code. The blue magnetic path operators $W_l^m$
    \eqref{eqn:magnetic_operator} and red electric path operators $W_l^e$ \eqref{eqn:electric_operator}
    act along the snake and the comb, respectively. We index $\{\bz_p\}_p$ as $\{\bz_j\}_{j = 1}^{L^2}$ along the snake from left to right and up to down, and index the spins from $2$ to $L^2$ on the snake in the same way so that each spin $j$ corresponds a pair of plaquettes $j - 1$ and $j$. One can similarly index the stars and spins on comb from left to right and up to down.
    }
    \label{fig:2d_snake_comb}
\end{figure}

\subsection{Spectral gap of the Gibbs sampler} \label{sec:gibbs_2d}

In this section, we will focus on steps \eqref{b} and \eqref{c} of the roadmap in \cref{sec:technical_overview}. For this, we decompose the generator \eqref{eqn:fast_mixing_Davies} as follows:
\begin{equation} \label{eq:sampler2d}
\mc{L}_\beta = \mc{L}_{\text{\rm local full}} + \mc{L}_{\rm global} = \mc{L}^{\rm gapped}+\mc{L}^{\rm rest}\,,
\end{equation}
where
\begin{equation}\label{eqn:L_gapped}
\begin{aligned}
    \mc{L}^{\rm gapped}:=\underbrace{\sum_{j\in \rm snake}\mathcal{L}_{\sigma^x_j}+\sum_{j\in \rm comb}\mathcal{L}_{\sigma^z_j}}_{:=\mathcal{L}_{\rm local}}+\underbrace{\mathcal{L}_{\xx_1}+\mathcal{L}_{\xx_2}+\mathcal{L}_{\zz_1}+\mathcal{L}_{\zz_2}}_{:=\mathcal{L}_{\rm global}}\,.
\end{aligned}
\end{equation}
Here $\mc{L}^{\rm rest}:= \mc{L}_{\text{\rm local full}} - \mc{L}_{\rm local}$  is a local Lindbladian with other Pauli couplings not included in $\mc{L}_{\rm local}$.
From \cref{lem1} (item 2), it suffices to limit our discussion to the part $\mc{L}^{\rm gapped}$ and show that it is primitive with the desired spectral gap lower bound $\max\left\{\exp(-\Or(\beta)), \Omega(N^{-2})\right\}$.


Before we proceed, we prepare the explicit formulations of the Lindbladians involved in $\mc{L}^{\rm gapped}$ for subsequent analysis. For $\si_j^x$ with $j \in {\rm snake}$, we have
\begin{equation*}
     e^{i t H^{\rm toric}} \si_j^{x} e^{- i t H^{\rm toric}} =  e^{- i t \left(\bz_{p'} + \bz_p \right)} \si_j^{x} e^{i t \left(\bz_{p'} + \bz_p\right)}\,,
\end{equation*}
where $p$ and $p'$ are the two plaquettes with $j = p \cap p'$ being the intersection site. The bond observable $-(\bz_{p'} + \bz_p)$ has eigenvalues $-2, 0, 2$ with the associated eigenprojections denoted by $P^-_j$, $P_j^0$, $P^+_j$, respectively, which can be represented as follows:
\begin{equation} \label{2eq:exp_proj}
    P_j^0 = \frac{1}{2} (I - \bz_{p'}\bz_p)\,,\q P_j^\pm = \frac{1}{4}(I \mp \bz_{p'})(I \mp \bz_{p})\,.
\end{equation}
We then compute the Fourier components of $\si_j^{x}$ ($j \in {\rm snake}$):
\begin{equation} \label{2eqn:a_j_def}
    \begin{aligned}
        e^{i t H^{\rm toric}} \si_j^{x} e^{- i t H^{\rm toric}}
         = e^{- 4 i  t} \underbrace{P_j^- \si_j^x  P_j^+}_{a_{j}} + e^{4 i t} \underbrace{P_j^+ \si_j^{x}  P_j^-}_{a_{j}^\dag} + \underbrace{P_j^0 \si_j^{x} P_j^0}_{a_{j}^0}\,,
    \end{aligned}
\end{equation}
due to $P_j^\pm \si_j^{x} P_j^0 = P^+_j \si_j^{x} P^+_j = P^-_j \si_j^{x} P^-_j = 0$, where $a_{j}$, $a_{j}^\dag$, $a_{j}^0$ correspond to the Bohr frequencies $-4$, $4$, $0$  of the Hamiltonian \eqref{2d:toric_ham}, respectively. Then, we can write $\mc{L}_{\si_j^{x}}$ by \cref{eqq:davies,eq:davies2} with Glauber-type transition \eqref{eqn:Glauber}:
\begin{multline}  \label{2eq:local_lind}
    \mc{L}_{\si_j^{x}}(A) = - \frac{1}{2} [a_{j}^0, [a_{j}^0, A]] + \frac{1}{2} \left\{ h_+ a_{j}^\dag [A, a_{j}] + h_- [a_{j}, A] a_{j}^\dag \right\} \\ + \frac{1}{2} \left\{h_+ [a_{j}^\dag ,A] a_{j} + h_- a_{j} [A, a_{j}^\dag] \right\}\,,
\end{multline}
where the constants $h_\pm$ are given by
\begin{equation} \label{eq:consthpm}
    h_+ =  \gamma(-4) = \frac{2}{e^{- 4 \beta} + 1}\,,\q h_- = \gamma(4) = \frac{2}{e^{4 \beta} + 1}\,.
\end{equation}

The computation for $\mc{L}_{\si_j^z}$ with $j \in {\rm comb}$ is quite similar to $\mc{L}_{\si_j^x}$ with $j\in {\rm snake}$, so we only sketch it below. For $\si_j^z$ with $j \in {\rm comb}$, we have
\begin{align*}
     e^{i t H^{\rm toric}} \si_j^{x} e^{- i t H^{\rm toric}} & =  e^{- i t \left(\bx_{s'} + \bx_s \right)} \si_j^z e^{i t \left(\bx_{s'} + \bx_s\right)} \\
     & = e^{- 4 i  t} \underbrace{Q_j^- \si_j^z  Q_j^+}_{b_{j}} + e^{4 i t} \underbrace{Q_j^+ \si_j^{z}  Q_j^-}_{b_{j}^\dag} + \underbrace{Q_j^0 \si_j^z Q_j^0}_{b_{j}^0}\,,
\end{align*}
with $s$ and $s'$ being two starts uniquely determined by $j = s \cap s'$. Here, the projections $Q^-_{j}$, $Q_{j}^0$, $Q^+_{j}$ are the eigenprojections of the bond observable $-(\bx_{s'} + \bx_s)$ for eigenvalues $-2, 0, 2$, which can be similarly formulated as \eqref{2eq:exp_proj} by replacing $\bz_p$ by $\bx_s$. With the Glauber-type transition rate \eqref{eqn:Glauber}, the generator $\mc{L}_{\si_j^{x}}$ is given by
\begin{multline}  \label{2eq:local_lindlz}
    \mc{L}_{\si_j^{z}}(A) = - \frac{1}{2} [b_{j}^0, [b_{j}^0, A]] + \frac{1}{2} \left\{ h_+ b_{j}^\dag [A, b_{j}] + h_- [b_{j}, A] b_{j}^\dag \right\} \\  + \frac{1}{2} \left\{h_+ [b_{j}^\dag ,A] b_{j} + h_- b_{j} [A, b_{j}^\dag] \right\}\,,
\end{multline}
where the constants $h_\pm$ is the same as \eqref{eq:consthpm}. According to our construction, we have
\begin{equation*}
\big\|\mc{L}_{\si_j^{z}}\big\|_{\si_\beta \to \si_\beta},\ \,  \big\|\mc{L}_{\si_j^{x}}\big\|_{\si_\beta \to \si_\beta} = \Theta(1)\,, \q \text{uniformly in $\beta$}\,.
\end{equation*}

We proceed to compute the Lindbladian with global couplings. For $\mf{O} = \xx_1,\xx_2,\zz_1,\zz_2$, thanks to $e^{i t H^{\rm toric}} \mf{O} e^{- i t H^{\rm toric}} = \mf{O}$ by $[\bz_p, \mf{O}] = [\bx_s, \mf{O}] = 0$, we readily have
\begin{equation} \label{2eq:lx}
    \mc{L}_{\mf{O}} (A) = - \frac{1}{2} [\mf{O}, [\mf{O}, A]]\,,
\end{equation}
since $\mf{O}$ has only the component with Bohr frequency zero. \REV{According to~\eqref{eq:supernorm} and \eqref{2eq:lx}, it is straightforward to see that
\begin{equation}  \label{2eq:lxnorm}
    \norm{\mc{L}_{\mf{O}}}_{\si_\beta \to \si_\beta} = \Theta(\|\mf{O}\|_{\rm op}^2)=\Theta(1)\,, \q \text{uniformly in $\beta$}\,,
\end{equation}
where $\norm{\dd}_{\rm op}$ denotes the standard operator norm.}


Now, for the step \eqref{b} of the roadmap in \cref{sec:technical_overview}, we first give the following lemma. Its proof is postponed to \cref{subsec:proof} for ease of reading.

\begin{lemma}\label{lem:L_gapped_property}
The generators $\mathcal{L}_{\rm global}$, $\sum_{j\in \rm snake}\mathcal{L}_{\sigma^x_j}$, and $\sum_{j\in \rm comb}\mathcal{L}_{\sigma^z_j}$ defined in \eqref{eqn:L_gapped}
only nontrivially act on a subalgebra of $\mc{B}(\mc{H})$. Specifically, we have
\begin{equation}\label{eqn:LXbar_commute}
\mathcal{L}_{\rm global}\left(\mathcal{Q}_1\otimes \mathcal{Q}_2\otimes \mathcal{A}^{\rm full}_{\rm m}\otimes \mathcal{A}^{\rm full}_{\rm e}\right)=\mathcal{L}_{\rm global}\left(\mathcal{Q}_1\otimes \mathcal{Q}_2\right)\otimes \mathcal{A}^{\rm full}_{\rm m}\otimes \mathcal{A}^{\rm full}_{\rm e}\,,
\end{equation}
\begin{equation}\label{eqn:Lxj_commute}
    \sum_{j\in \rm snake}\mathcal{L}_{\sigma^x_j}\left(\mathcal{Q}_1\otimes \mathcal{Q}_2\otimes \mathcal{A}^{\rm full}_{\rm m}\otimes \mathcal{A}^{\rm full}_{\rm e}\right)=\mathcal{Q}_1\otimes \mathcal{Q}_2\otimes \sum_{j\in \rm snake}\mathcal{L}_{\sigma^x_j}\left(\mathcal{A}^{\rm full}_{\rm m}\right)\otimes \mathcal{A}^{\rm full}_{\rm e}\,,
\end{equation}
and
\begin{equation}\label{eqn:Lzj_commute}
\sum_{j\in \rm comb}\mathcal{L}_{\sigma^z_j}\left(\mathcal{Q}_1\otimes \mathcal{Q}_2\otimes \mathcal{A}^{\rm full}_{\rm m}\otimes \mathcal{A}^{\rm full}_{\rm e}\right)=\mathcal{Q}_1\otimes \mathcal{Q}_2\otimes \mathcal{A}^{\rm full}_{\rm m}\otimes \sum_{j\in \rm comb}\mathcal{L}_{\sigma^z_j}\left(\mathcal{A}^{\rm full}_{\rm e}\right)\,.
\end{equation}
In particular, it holds that
\begin{equation}\label{eqn:kernel}
\begin{aligned}
\mathrm{Ker}\Big(\sum_{j\in \rm snake}\mathcal{L}_{\sigma^x_j}\Big)=&~\mathcal{Q}_1\otimes \mathcal{Q}_2\otimes {\rm \mi}\otimes \mathcal{A}^{\rm full}_{\rm e}\,,\\
\mathrm{Ker}\Big(\sum_{j\in \rm comb}\mathcal{L}_{\sigma^z_j}\Big)=&~\mathcal{Q}_1\otimes \mathcal{Q}_2\otimes \mathcal{A}^{\rm full}_{\rm m}\otimes{\rm \mi}\,,\\
\mathrm{Ker}\left(\mathcal{L}_{\rm global}\right)=&~{\rm \mi}\otimes {\rm \mi}\otimes \mathcal{A}^{\rm full}_{\rm m}\otimes \mathcal{A}^{\rm full}_{\rm e}\,.
\end{aligned}
\end{equation}
\end{lemma}

Note from \eqref{2eq:lx} that for any $B_1,B_2\in \{I, \xx_i, \yy_i, \zz_i\}$, the operator $B_1 B_2$ is an eigenvector of $\mathcal{L}_{\rm global}$: $\mathcal{L}_{\rm global}\left(B_1 B_2\right) = c \cdot B_1 B_2$ for some constant $c$. Combining this with \cref{lem:L_gapped_property}, we obtain the following properties of $\mc{L}^{\rm gapped}$.

\begin{cor}\label{lem:L_gapped_property2}
It holds that
\begin{itemize}
\item $\mc{L}^{\rm gapped}$ is block diagonal for the following orthogonal decomposition \eqref{eqn:A_full_decomposition} in both HS inner product and GNS inner product:
\begin{equation}\label{eqn:A_full_decomposition}
\mathcal{Q}_1\otimes \mathcal{Q}_2\otimes \mathcal{A}^{\rm full}_{\rm m}\otimes \mathcal{A}^{\rm full}_{\rm e}=\bigoplus_{B_i\in \{I, \xx_i, \yy_i, \zz_i\},~ i=1,2} \mc{B}_{B_1,B_2}\,,
\end{equation}
with
\begin{equation*}
\mc{B}_{B_1,B_2} = B_1\otimes B_2\otimes \mathcal{A}^{\rm full}_{\rm m}\otimes \mathcal{A}^{\rm full}_{\rm e}\,.
\end{equation*}

\item $\mc{L}^{\rm gapped}$ is primitive: $\mathrm{Ker}(\mc{L}^{\rm gapped})= {\rm Span}\{\mi\}$.

\end{itemize}
\end{cor}

By \eqref{eq2} in \cref{lem1} with the primitivity of $\mathcal{L}^{\rm gapped}$, we have
\begin{equation*}
    \mathrm{Gap}(-\mc{L}_\beta)=\mathrm{Gap}(-\mc{L}^{\rm rest}-\mc{L}^{\rm gapped})\geq \mathrm{Gap}(-\mc{L}^{\rm gapped})\,.
\end{equation*}
Further, from the block diagonal form of  $\mc{L}^{\rm gapped}$ for \eqref{eqn:A_full_decomposition}, we obtain \eqref{eqn:Gap_inequality}:
\begin{multline} \label{eqn:Gap_inequality_later}
\mathrm{Gap}\left(-\mc{L}_\beta\right)\geq \mathrm{Gap}\left(-\mc{L}^{\rm gapped}\right) \\ \geq \min\left\{\mathrm{Gap}\left(-\mc{L}^{\rm gapped}\middle|_{\mc{B}_{\mi,\mi}}\right),\lambda_{\min}\left(-\mc{L}^{\rm gapped}|_{\mc{B}_{B_1,B_2}}\right)~\text{for}~B_1\neq \mi\ \text{or}\ B_2\neq \mi\right\}\,.
\end{multline}

For the first term in \eqref{eqn:Gap_inequality_later}, we note that $\mc{L}_{\rm global}|_{\mc{B}_{\mi,\mi}} = 0$ from \eqref{eqn:kernel}, and then have
\begin{equation}  \label{eqn:gapfirst}
    \mathrm{Gap}\left(-\mc{L}^{\rm gapped}\middle|_{\mc{B}_{\mi,\mi}}\right) = \mathrm{Gap}\left(-\mc{L}_{\rm local}\middle|_{\mc{B}_{\mi,\mi}}\right)\,.
\end{equation}

We next consider the lower bound estimation of $\lambda_{\min}\left(-\mathcal{L}^{\rm gapped}|{\mathcal{B}_{B_1,B_2}}\right)$. For any $B_i \in \{I, \xx_i, \yy_i, \zz_i\}$ with $i = 1,2$  such that $B_1 \neq \mathbf{1}$ or $B_2 \neq \mathbf{1}$, the operators $\xx_1, \xx_2, \zz_1, \zz_2$ either commute or anti-commute with $B_1B_2$, and there always exists one, say $\mf{O}$, among them that anti-commutes with $B_1 B_2$. Then, by \eqref{2eq:lx}, we have $\mathcal{L}_{\mf{O}}(A) = -2A$, which implies
\begin{equation}\label{eqn:L_global_nonvanish}
\left\langle B_1B_2,-\mathcal{L}_{\rm global}(B_1B_2)\right\rangle_{\sigma_\beta} \ge \left\langle B_1B_2,-\mathcal{L}_{\mf{O}}(B_1B_2)\right\rangle_{\sigma_\beta} = 2\,.
\end{equation}
Furthermore, according to \cref{eqn:Lxj_commute,eqn:Lzj_commute,eqn:kernel} in \cref{lem:L_gapped_property}, we find
\[
\mathrm{Ker}\left(\mc{L}_{\rm local}\middle|_{\mc{B}_{B_1,B_2}}\right)=B_1B_2\,.
\]
Finally, by \cref{lem1} (item 4), there holds
 \begin{equation}\label{eqn:before_prop}
\begin{aligned}
        - \mc{L}^{\rm gapped}\Big|_{\mc{B}_{B_1,B_2}} &\succeq \frac{\gap\left(\mc{L}_{\rm local}\middle|_{\mc{B}_{B_1,B_2}}\right) \left\langle B_1B_2,-\mathcal{L}_{\rm global}(B_1B_2)\right\rangle_{\sigma_\beta}}{\gap\left(\mc{L}_{\rm local}\middle|_{\mc{B}_{B_1,B_2}}\right) + \norm{\mc{L}_{\rm global}}_{\sigma_\beta\rightarrow \sigma_\beta}} \\
        &= \Theta\left(\frac{\gap\left(\mc{L}_{\rm local}\middle|_{\mc{B}_{B_1,B_2}}\right)}{\gap\left(\mc{L}_{\rm local}\middle|_{\mc{B}_{B_1,B_2}}\right) + 1}\right)\,,
      \end{aligned}
\end{equation}
where we have used \eqref{2eq:lxnorm}.
Plugging this into \eqref{eqn:Gap_inequality_later}, we obtain the following proposition, concluding step \eqref{b} of the proof.
\begin{prop}\label{prop:second_step}
Let $\mc{L}_\beta$, $\mc{L}^{\rm gapped}$, $\mc{L}_{\rm global}$ and $\mc{L}_{\rm local}$ be defined in \eqref{eq:sampler2d}-\eqref{eqn:L_gapped}. Suppose $\mathrm{Gap}\left(\mc{L}_{\rm local}\middle|_{\mc{B}_{\mi,\mi}}\right) = \Or(1)$. Then, for any $B_i \in \{I, \xx_i, \yy_i, \zz_i\}$ with $i = 1,2$  such that $B_1 \neq \mathbf{1}$ or $B_2 \neq \mathbf{1}$, we have
\begin{equation}\label{eqn:Gap_inequality_later_2}
  -\mc{L}^{\rm gapped}|_{\mc{B}_{B_1,B_2}}\succeq \Omega\left(\mathrm{Gap}\left(\mc{L}_{\rm local}\middle|_{\mc{B}_{\mi,\mi}}\right)\right)\,,
\end{equation}
Moreover, it holds that
\begin{equation}\label{eqn:Gap_inequality_later_3}
\mathrm{Gap}\left(\mc{L}_\beta\right)=\Omega\left(\mathrm{Gap}\left(-\mc{L}_{\rm local}\middle|_{\mc{B}_{\mi,\mi}}\right)\right)\,.
\end{equation}
\end{prop}
\begin{proof}
Since the action of $\mc{L}_{\rm local}$ on $\mc{B}_{B_1, B_2}$ is independent of $B_1,B_2$ by \eqref{eqn:Lxj_commute} and \eqref{eqn:Lzj_commute}, we obtain $\mathrm{Gap} (\mc{L}_{\rm local}|_{\mc{B}_{B_1,B_2}})=\mathrm{Gap}(\mc{L}_{\rm local}|_{\mc{B}_{\mi,\mi}})$,
which, along with \eqref{eqn:before_prop}, implies \eqref{eqn:Gap_inequality_later_2}. Then, the estimate \eqref{eqn:Gap_inequality_later_3} directly follows from \eqref{eqn:Gap_inequality_later} and \eqref{eqn:gapfirst}.
\end{proof}

Thanks to \cref{prop:second_step} above, to finish the proof of \cref{thm:fast_mixing_2D}, it suffices to study the spectral gap of $\mc{L}_{\rm local}$ on the syndrome space $\mc{B}_{\mi,\mi}$, i.e., step \eqref{c} outlined in \cref{sec:technical_overview}. This is the goal of the next section.



\subsubsection{Analysis of the quasi-1D structure}\label{sec:gap_quasi_1D}

In this section, we will analyze the spectral gap of the local Davies generator $\mc{L}_{\rm local}$ in the syndrome space. The main result is stated as follows.

\begin{prop}\label{prop:gap_quasi_1D}
Let $\mc{L}_{\rm local}$ be defined in \eqref{eqn:L_gapped} and $\mc{B}_{\mi,\mi}$ be defined in \cref{lem:L_gapped_property2}, we have
\[
\mathrm{Gap}\left(-\mc{L}_{\rm local}\middle|_{\mc{B}_{\mi,\mi}}\right) = \max\left\{e^{-\Or(\beta)}, \Omega(N^{-2})\right\}\,.
\]
\end{prop}
Then, \cref{thm:fast_mixing_2D} is a corollary of \cref{prop:second_step} and \cref{prop:gap_quasi_1D}. To prove the above proposition, we consider
\begin{equation}\label{eq:locallind_2}
\mc{L}_1:=\sum_{j\in \rm snake}\mathcal{L}_{\sigma^x_j},\quad \mc{L}_2:=\sum_{j\in \rm comb}\mathcal{L}_{\sigma^z_j}\,.
\end{equation}
From \cref{lem:L_gapped_property}, it is straightforward to see that $\mc{L}_1$ and $\mc{L}_2$ commutes and only act nontrivally on $\mc{A}^{\rm full}_{\rm m}$ and $\mc{A}^{\rm full}_{\rm e}$, respectively. Thus, we can analyze the spectral gap of $\mc{L}_1$ and $\mc{L}_2$ separately. \cref{prop:gap_quasi_1D} directly follows from the following two lemmas.

\begin{lemma}\label{lem:L_1}
Let $\mc{A}^{\rm full}_{\rm m}$ be defined as in \cref{lem:2D_decomposition}, we have
\begin{equation}\label{eqn:Gap_L_1}
  \mathrm{Gap}\left(-\mc{L}_1\middle|_{\mc{A}^{\rm full}_{\rm m}}\right) = \max\left\{e^{-\Or(\beta)}, \Omega(N^{ - 2})\right\}\,.
\end{equation}
\end{lemma}
\begin{lemma}\label{lem:L_2} Let $\mc{A}^{\rm full}_{\rm e}$ be defined as in \cref{lem:2D_decomposition}, we have
  \[
    \mathrm{Gap}\left(-\mc{L}_2\middle|_{\mc{A}^{\rm full}_{\rm e}}\right) = \max\left\{e^{-\Or(\beta)}, \Omega(N^{-2})\right\}\,.
  \]
\end{lemma}


The $\exp(- \Or(\beta))$ spectral gap of $\mathrm{Gap}(-\mc{L}_1|_{\mc{A}^{\rm full}_{\rm m}})$ and $\mathrm{Gap}(-\mc{L}_2|_{\mc{A}^{\rm full}_{\rm e}})$ has been proved in \cite[Section 7]{Alicki_2009}. In what follows, we focus on the gap estimate of order ${\rm poly}(N^{-1})$. We first prove \cref{lem:L_1}.

\begin{proof}[Proof of \cref{lem:L_1}]
Recall that $\mc{A}^{\rm full}_{\rm m}=\mc{B}\left(\mc{H}^{\rm m}_{\rm b}\right)$ is spanned by the basis matrices $\ket{m}\bra{m'}$, where $m,m'\in \{-1,1\}^{L^2}$ and $\#\left\{m_i=-1\right\}, \#\left\{m'_i=-1\right\} \in 2\ZZ$. We note that this set of basis matrices $\left\{\ket{m}\bra{m'}\right\}$ also forms an orthogonal basis for $\mc{A}^{\rm full}_{\rm m}$ with respect to the GNS inner product defined as in \eqref{eqn:GNS_inner_product} by the reduced Gibbs state $\tr_{\mc{A}^{\rm full}_{\rm e}}(\si_\beta) \propto \exp(- \beta \sum_p \bz_p)$
\footnote{Note from \cref{lem:2D_decomposition} that $\si_\beta \in \mathcal{A}^{\rm full}_{\rm m} \otimes \mathcal{A}^{\rm full}_{\rm e}$.}.


We next decompose the space $\mathcal{A}^{\rm full}_{\rm m}$ such that $\mc{L}_1$ presents a block diagonal form.  Note that $\mathcal{A}^{\rm full}_{\rm m}$ is spanned by $\ket{m'}\bra{m}$ with $\ket{m},\ket{m'} \in \mc{H}^{\rm m}_{\rm b}$ having an even number of $-$ signs.
For any $\Lambda \subset \{1,2,\ldots,L^2\}$ with even $L^2 - |\Lambda|$, we introduce
the subspace $\mathcal{A}^{\rm full}_{\rm m}(\Lambda)$ of $\mathcal{A}^{\rm full}_{\rm m}$
spanned by $\ket{m'}\bra{m}$, where
$m = m'$ on $\Lambda$ and $m =  - m'$ on the complement $\Lambda^c := \{1,2,\ldots,L^2\} \backslash \Lambda$. It is straightforward to check that the subspaces $\mathcal{A}^{\rm full}_{\rm m}(\Lambda)$ are
orthogonal with respect to both the GNS and HS inner products, and thereby
\begin{align} \label{eq:decom}
  \mc{A}^{\rm full}_{\rm m} =\bigoplus_{\substack{\Lambda\subset \{1,\dots,L^2\}:\\L^2-\Lambda~\text{even}}} \mathcal{A}^{\rm full}_{\rm m}(\Lambda)\,.
\end{align}
In addition, by writing $\ket{m'}\bra{m} = (\ket{m}\bra{m})_{\Lambda} \otimes (\ket{-m}\bra{m})_{\Lambda^c} \in \mathcal{A}^{\rm full}_{\rm m}(\Lambda)$,
we can further decompose
\begin{align*}
   \mathcal{A}^{\rm full}_{\rm m}(\Lambda) = \mc{A}^{\rm ab}(\Lambda) \otimes \mc{F}(\Lambda)\,,
\end{align*}
where $\mc{A}^{\rm ab}(\Lambda)$ is the Abelian algebra generated by the projections on bonds restricted to $\Lambda$ and $\mc{F}(\Lambda)$ is the space spanned by flips of bonds restricted to $\Lambda^c$.
We observe that any partition of $\{1,2,\ldots, L^2\}$ into $\Lambda \cup \Lambda^c$ induces a partition of the spins $j$ on the snake,  equivalently, pairs of neighboring plaquettes/bonds $\{\bz_{j-1},\bz_{j}\}$
into three sets: $\Gamma_{\rm flip}$ for both bonds in $\Lambda^c$, $\Gamma_{\rm ab}$ for both bonds in $\Lambda$, and $\Gamma_{\rm int}$ for one bond in $\Lambda$ and the other in $\Lambda^c$.

The following lemma extends \cite[Lemma 5]{Alicki_2009}, which was originally stated for the 1D ferromagnetic Ising model. We provide a self-contained proof in \cref{subsec:proof} with explicit computations of the local matrix representations of the master Hamiltonian of $\mc{L}_1$.

\begin{lemma} \label{lem:blockdecom}
    $\mc{L}_1$, defined in \eqref{eq:locallind_2}, on $\mathcal{A}^{\rm full}_{\rm m}$ is block diagonal for the decomposition \eqref{eq:decom}.
\end{lemma}

Now, we are ready to estimate the gap of $\mc{L}_1$ on each $\mc{A}^{\rm full}_{\rm m}(\Lambda)$ by the following three cases.


\begin{itemize}
    \item If $\Gamma_{\rm int} \neq \emptyset$, by \cref{matrix_int} in the    proof of \cref{lem:blockdecom}, under an orthonormal basis with respect to the GNS inner product, the master Hamiltonian of $- \mc{L}_{\sigma^x_j}$ on $\mc{A}^{\rm full}_{\rm m}(\Lambda)$ takes the form:
    \begin{equation} \label{matrix_int_main}
    I_{j-2}\otimes \left[\begin{matrix}
     \frac{h_- + 1}{2} & 0 & 0 & 0 \\
    0 &  \frac{h_+ + 1}{2} & 0 & 0 \\
    0 & 0 &  \frac{h_- + 1}{2} & 0   \\
    0 & 0 & 0 &  \frac{h_+ + 1}{2}\\
    \end{matrix}\right] \otimes  I_{L^2 - j}\succeq \frac{1}{2}\,.
    \end{equation}
    This implies
    \begin{align*}
        - \left\l A, \mc{L}_1|_{\mc{A}^{\rm full}_{\rm m}(\Lambda)} A \right \r_{\si_\beta} \ge \sum_{j \in \Gamma_{\rm int}} - \left\l A, \mc{L}_{\sigma^x_j}|_{\mc{A}^{\rm full}_{\rm m}(\Lambda)} A \right \r_{\si_\beta} \ge \left|\Gamma_{\rm int}\right|\cdot \frac{1}{2} \l A, A \r_{\si_\beta}\,.
    \end{align*}
    for any $A\in \mc{A}^{\rm full}_{\rm m}(\Lambda)$, where the first step is by $-\mc{L}_{\sigma_j^x} \succeq 0$. It follows that
    \begin{equation}\label{eqn:gap_int}
        - \mc{L}_1|_{\mc{A}^{\rm full}_{\rm m}(\Lambda)} \succeq  \frac{\left|\Gamma_{\rm int}\right|}{2}\,.
    \end{equation}
    \item If $\Gamma_{\rm flip} = \{1,2,\ldots,L^2\}$ (i.e., $\Lambda=\emptyset$), by \eqref{matrix_flip}, the master Hamiltonian of $- \mc{L}_{\sigma^x_j}$ on $\mc{A}^{\rm full}_{\rm m}(\emptyset)$ is
     \begin{equation} \label{matrix_flip_main}
        I_{j-2}\otimes \left[\begin{matrix}
           - \frac{h_+ + h_-}{2} & 0 & 0 & 0 \\
            0 & -1 & 1 & 0 \\
            0 & 1 & -1 & 0   \\
            0 & 0 & 0 & - \frac{h_+ + h_-}{2}\\
        \end{matrix}\right] \otimes I_{L^2 - j}\,.
    \end{equation}
    This case has already been discussed in~\cite{Alicki_2009} (see after Proposition 2), from which we have
    \begin{equation}\label{eqn:gap_flip}
        -\mc{L}_1|_{\mc{A}^{\rm full}_{\rm m}(\emptyset)}\succeq \frac{1}{2}\,.
    \end{equation}

    \item  If $\Gamma_{\rm ab} = \{1,2,\ldots,L^2\}$ (i.e., $\Lambda=[L^2]$),
    \cite[Proposition 2]{Alicki_2009} has given a spectral gap lower bound that exponentially decays in $\beta$. However, this lower bound is far from sharp at low temperatures. Next, we shall derive a lower bound of gap polynomially decaying in $N$ but independent of $\beta$, which is summarized in the following lemma.
    \begin{lemma}\label{lem:gap_abelian} Given the notation above, we have,  when $\beta = \Omega(\ln N)$,
    \[
    \mathrm{Gap}\left(-\mc{L}_1\middle|_{\mc{A}^{\rm full}_{\rm m}\left([L^2]\right)}\right) = \Omega(N^{-2})\,.
    \]
    \end{lemma}
\end{itemize}
The proof of \cref{lem:gap_abelian} is postponed to \cref{subsec:proof} for readability, \REV{which relies on a novel estimate of the minimal eigenvalue of a perturbed graph Laplacian.}
Once \cref{lem:gap_abelian} is proved, we can combine it with \cref{lem:blockdecom}, as well as
\eqref{eqn:gap_int} and \eqref{eqn:gap_flip}, to obtain the desired \eqref{eqn:Gap_L_1}:
\begin{align*}
\mathrm{Gap}\left(-\mc{L}_1|_{\mc{A}_{\rm m}^{\rm full}}\right)\geq \min\left\{\frac{1}{2},\mathrm{Gap}\left(-\mc{L}_1\middle|_{\mc{A}^{\rm full}_{\rm m}\left([L^2]\right)}\right)\right\}= \Omega(N^{-2})\,.
\end{align*}
\end{proof}

We next prove \cref{lem:L_2}, whose basic ideas are similar to that of \cref{lem:L_1} but require some more technical arguments due to the comb structure. To be specific, recall the formula \eqref{2eq:local_lindlz}, the local master Hamiltonian representation of $\mc{L}_{\si_j^{z}}$ is the same as that of $\mc{L}_{\si_j^{x}}$. However, here $\mc{L}_1$ acts on a 1D straight line (snake), where each observable $\bz_p$ connects only two neighboring qubits along the chain, while $\mc{L}_2$ acts on a comb-like 1D structure, where some observables $\bx_s$ can connect three neighboring qubits on the comb. This difference prevents us from directly applying the proof of \cref{lem:L_1}.


\begin{proof}[Proof of \cref{lem:L_2}]

The starting point of the proof of \cref{lem:L_2} follows from that of \cref{lem:L_1}. We recall that $\mc{A}^{\rm full}_{\rm e}$ is spanned by the orthogonal basis $\ket{e'}\bra{e}$ for the GNS inner product induced by the reduced Gibbs state $\tr_{\mc{A}^{\rm full}_{\rm m}}(\si_\beta) \propto \exp(- \beta \sum_s \bx_s)$,
where $e,e'\in \{+,-\}^{L^2}$ and $\#\left\{e_i=-1\right\}, \#\left\{e'_i=-1\right\} \in 2 \ZZ$. Moreover, we decompose $\mc{A}^{\rm full}_{\rm e} =\bigoplus_{\Lambda} \mathcal{A}^{\rm full}_{\rm e}(\Lambda)$, where $\Lambda \subset \{1,2,\ldots,L^2\}$ with even $L^2 - |\Lambda|$ and the subspace $\mathcal{A}^{\rm full}_{\rm e}(\Lambda)$ is spanned by $\ket{e'}\bra{e}$, where
$e = e'$ on $\Lambda$ and $e =  - e'$ on $\{1,2,\ldots,L^2\} \backslash \Lambda$. We also partition pairs of neighboring bonds into three sets: $\Gamma_{\rm flip}$ for both bonds in $\Lambda^c$, $\Gamma_{\rm ab}$ for both bonds in $\Lambda$, and $\Gamma_{\rm int}$ for one bond in $\Lambda$ and the other in $\Lambda^c$.

Then, a similar lemma as \cref{lem:blockdecom} holds for $\mc{L}_2$, since its proof only needs the properties of local Lindbladians $\mc{L}_{\si_j^z}$ that are the same as those of $\mc{L}_{\si_j^x}$.

\begin{lemma}
    $\mc{L}_2$, defined in \eqref{eq:locallind_2}, on $\mathcal{A}^{\rm full}_{\rm e}$ is block diagonal for the decomposition:
    \begin{equation*}
         \mc{A}^{\rm full}_{\rm e} =\bigoplus_{\substack{\Lambda\subset \{1,\dots,L^2\}:\\L^2-\Lambda~\text{even}}} \mathcal{A}^{\rm full}_{\rm e}(\Lambda)\,.
    \end{equation*}
\end{lemma}

Next, we consider three cases: 1. $\Gamma_{\rm int}\neq \emptyset$; 2. $\Gamma_{\rm flip}= \{1,\ldots,L^2\}:=[L^2]$; 3. $\Gamma_{\rm ab} = [L^2]$. Noting that again, the arguments of \eqref{eqn:gap_int} and \eqref{eqn:gap_flip} only uses
local Lindbladians $\mc{L}_{\si_j^x}$, we have similar estimates for $\mc{L}_2$:
\[
-\mathcal{L}_2|_{\mc{A}_{\rm e}^{\rm full}(\Lambda),\,\Gamma_{\rm int}\neq \emptyset} \succeq \frac{1}{2}\quad \text{and}\quad -\mathcal{L}_2|_{\mc{A}_{\rm e}^{\rm full}(\emptyset)} \succeq \frac{1}{2}\,.
\]
We now consider the third case $\Gamma_{\rm ab} = [L^2]$ and prove the following result, which finishes the proof of \cref{lem:L_2}.

\begin{lemma}\label{lem:gap_abelian_comb} Given the notation above, we have, when $\beta = \Omega(\ln N)$,
    \[
    \mathrm{Gap}\left(-\mc{L}_2\middle|_{\mc{A}^{\rm full}_{\rm m}\left([L^2]\right)}\right)= \Omega(N^{-2})\,.
    \]
\end{lemma}
The proof of \cref{lem:gap_abelian_comb} is deferred to \cref{subsec:proof}.
\end{proof}

\begin{figure}[bthp]
\includegraphics[width=\textwidth]{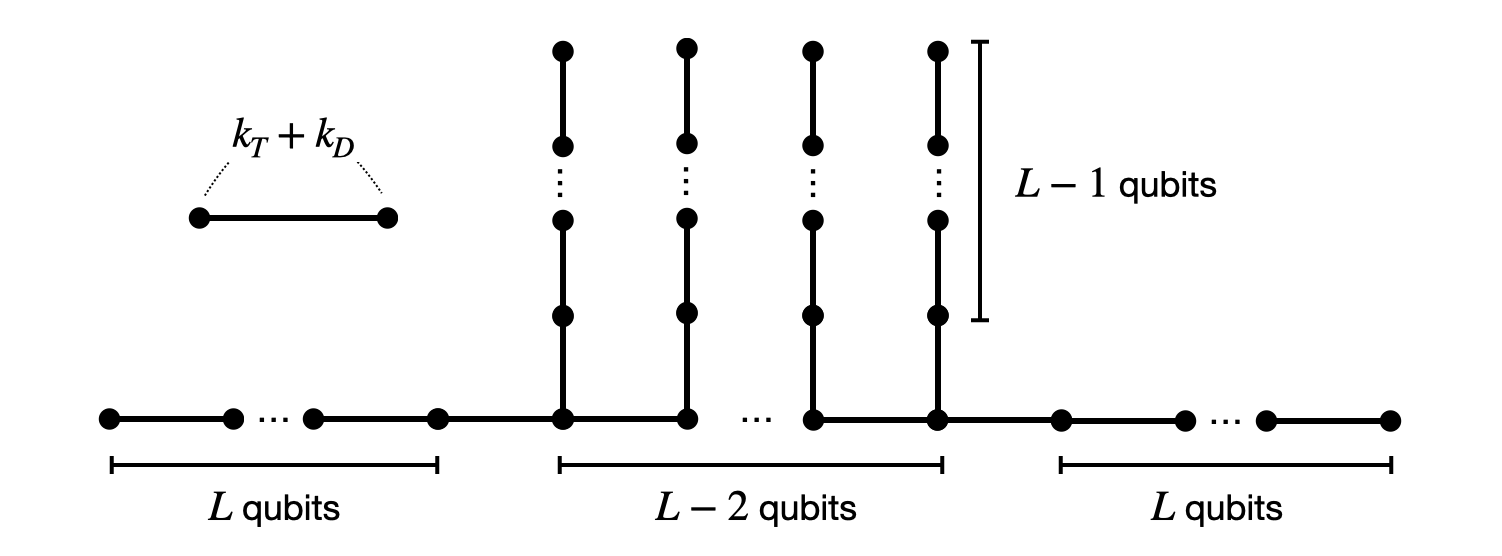}
\caption{The underlying graph of $\sum_{j\in \rm comb}\mathcal{L}_{\sigma^z_j}$ acting on $\mathcal{A}^{\rm full}_{\rm e}$. Here, each black dot represents a qubit corresponding to a bond observation of $\bx_{s_j}$. The transition matrix $k_T+k_D$ induced by $\mc{L}_{\sigma^z_j}$ acts on each neighboring qubits.}
\label{fig:2d1}
\end{figure}
\begin{figure}[bthp]
\includegraphics[width=\textwidth]{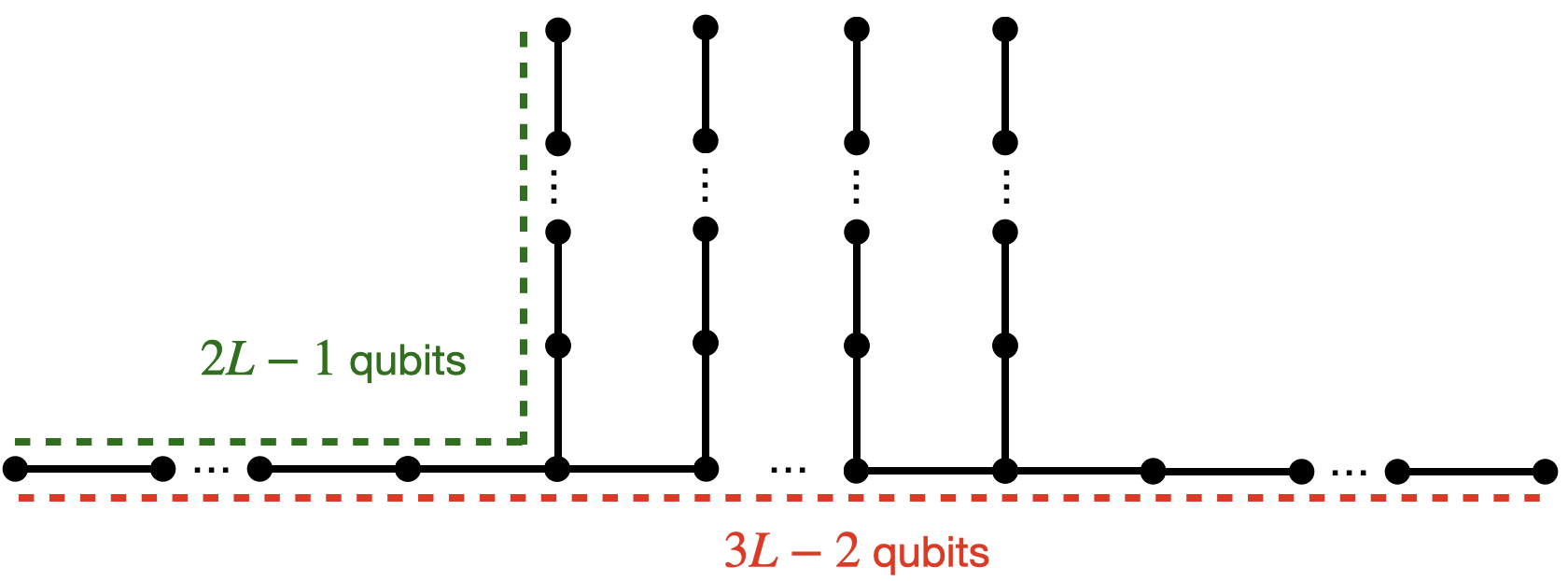}
\caption{Examples of lines on the graph. Both lines are the shortest path between two
degree-one vertices (i.e., end-points) on the comb-like graph. The detailed application of these lines can be found in the proof of \cref{lem:gap_abelian_comb}.
}
\label{fig:2d2}
\end{figure}

\subsection{Proof of Lemmas} \label{subsec:proof}

We collect proofs of some technical lemmas for the spectral gap analysis.

\begin{proof}[Proof of \cref{lem:2D_decomposition}]
The decomposition \eqref{eq:opdecom} is straightforward by the construction. Here we only prove the fact that the algebra $\mc{A}^{\rm full}_{\rm m} = \mc{B}(\mc{H}_{\rm b}^{\rm m})$ can be generated by $\{\mathbf{Z}_p\}_p\cup \{\sigma^x_j\}_{j\in \text{\rm snake}}$. The claim for $\mc{A}^{\rm full}_{\rm e}$ can be similarly proved. Indeed, each basis (admissible bond) $\ket{m}$ in $\mathcal{H}_{\rm b}^{\rm m}$ can be identified as a Pauli string $\sigma^x_{j_1}\cdots \sigma^x_{j_k}$ with $j_i$ on the snake. There are different ways to index the plaquette: and spins: One example is to index $\{\bz_p\}_p$ as $\{\bz_j\}_{j = 1}^{L^2}$ along the snake from left to right and up to down, and index the spins from $2$ to $L^2$ on the snake in the same way so that each spin $j$ corresponds a pair of plaquettes $j - 1$ and $j$, see~\cref{fig:2d_snake_comb} for detail. Each operator on $\mathcal{H}_{\rm b}^{\rm m}$
can be written as a linear combination of $\ket{m'}\bra{m}$, where $\ket{m}$ and $\ket{m'}$ are admissible bonds. Meanwhile, each $\ket{m'}\bra{m}$ with $\ket{m'} \neq \ket{m}$ is a composition of flips of neighboring states. Thus, it suffices to consider the case
$\ket{m'}\bra{m}$ that maps the state $\ket{m} = \sigma^x_{j_1}\cdots \sigma^x_{j_k} \ket{+1^{L^2}}$
 to a neighboring one $\ket{m'} = \sigma^x_{j_2}\cdots \sigma^x_{j_k} \ket{+1^{L^2}}$,
which can be constructed using $\{\bz_p\}_p$ and $\{\sigma^x_j\}_{j \in {\rm snake}}$. Here $\ket{m} = \sigma^x_{j_1}\cdots \sigma^x_{j_k} \ket{+1^{L^2}}$ is defined as the bond configuration for $\{\bz_p\}_p$ associated with the spin configuration $\sigma^x_{j_1}\cdots \sigma^x_{j_k} \ket{\psi_o}$, where $\ket{\psi_o}$ is a ground state \eqref{eqn:2D_basis}.
This gives the alternative representation for $\ket{m'}\bra{m}$:
\[
\ket{m'}\bra{m} =   \frac{I - m_{j_1}\bz_{j_1}}{2} \cdot \underbrace{\Big(\prod_{j \neq j_1 - 1, j_1}^{L^2} \frac{I + m_j \bz_j}{2}\Big)  \si^x_{j_1}}_{=\ket{m'}\bra{m} + \ket{m}\bra{m'}}\,.
\]
Here, $\frac{I - m_{j_1}\bz_{j_1}}{2}\ket{m'}=\ket{m'}$ and $ \frac{I - m_{j_1}\bz_{j_1}}{2}\ket{m}=0$. In the case when $j_1=1$, we ignore the condition $j\neq j_1-1$. The case of $\ket{m}\bra{m}$ can be similarly done.
\end{proof}

\begin{proof}[Proof of \cref{lem:L_gapped_property}] The formula \eqref{eqn:LXbar_commute} follows from the representation of $\mc{L}_{\ms{O}}$ with $\mf{O} = \xx_1,\xx_2,\zz_1,\zz_2$ and the fact from \eqref{lem:2D_decomposition} that these global jumps commute with the algebras $\mc{A}_{\rm m/e}^{\rm full}$. For the formula \eqref{eqn:Lxj_commute}, it suffices to note that $\{\si_j^x\}_{j \in {\rm snake}}$ and the projections \eqref{2eq:exp_proj} belong to $\mc{A}_{\rm m}^{\rm full}$ and commute with $\mc{Q}_1 \otimes \mc{Q}_2 \otimes \mc{A}_{\rm e}^{\rm full}$. The formula \eqref{eqn:Lzj_commute} follows from the same reason by the computation \eqref{2eq:local_lindlz}. The first two statements of \eqref{eqn:kernel} can be proved by a very similar argument as \cite[Lemma 6]{Alicki_2009}. To show $\mathrm{Ker}\left(\mathcal{L}_{\rm global}\right)=\mi\otimes \mi\otimes \mathcal{A}^{\rm full}_{\rm m}\otimes \mathcal{A}^{\rm full}_{\rm e}$, we only need to note that operators $\xx_1, \xx_2, \zz_1, \zz_2$ span the whole algebra $\mc{Q}_1 \otimes \mc{Q}_2$.
\end{proof}

\begin{proof}[Proof of \cref{lem:blockdecom}]
We first recall that we order the plaquette observables $\{\bz_p\}_p$ as $\{\bz_j\}_{j = 1}^{L^2}$ and index $\{\si_j^x\}_{j = 2}^{L^2}$
along the snake. We note from \eqref{2eq:exp_proj} that the projections, as elements in $\mc{A}^{\rm full}_{\rm m}$, can be represented as
\begin{equation*}
  P_j^0 = \left(\ket{+-}\bra{+-}+\ket{-+}\bra{-+}\right)_{j-1,j},\ P_j^\pm = \left(\ket{\mp}\bra{\mp}\right)_{j-1,j}\,,
\end{equation*}
where the subscript means that the operator only nontrivially acts on bonds $j - 1$ and $j$ on the snake associated with observables $\{\bz_p\}_p$.
This enables us to compute the jumps $a_{j}, a_{j}^\dag, a_{j}^0 \in \mc{A}^{\rm full}_{\rm m}$ defined in  \eqref{2eqn:a_j_def} for the local Lindbladian
$\mc{L}_1$ in \eqref{eq:locallind_2}:
\begin{align*}
    a_{j} = \left(\ket{++}\bra{--}\right)_{j-1,j}\,, \q  a_{j}^\dag = \left(\ket{--}\bra{++}\right)_{j-1,j}\,,\,,
\end{align*}
and
\begin{equation*}
    a_{j}^0 =  \left(\ket{+-}\bra{-+} + \ket{-+}\bra{+-}\right)_{j-1,j}\,.
\end{equation*}
Without loss of generality, we consider the Lindbladian $\mc{L}_{\sigma^x_j}$ in \eqref{2eq:local_lind} for a fixed $j \in \Gamma_{\rm flip}, \Gamma_{\rm ab}, \Gamma_{\rm int}$ on $\mc{A}^{\rm full}_{\rm m}(\Lambda)$ for each $\Lambda$. Due to the locality of the jump operators, $\mc{L}_{\sigma^x_j}$ only changes the pair of bonds associated with $j$.
For simplicity, we shall omit the subscripts $j-1,j$.

\begin{itemize}
\item For $j \in  \Gamma_{\rm ab}$, we consider the local basis
\begin{align*}
    A_1 = \ket{++}\bra{++}\,,\q A_2 = \ket{+ - }\bra{+ -}\,,\q A_3 = \ket{- +}\bra{- +},\q  A_4 = \ket{--}\bra{--}\,,
\end{align*}
that are orthogonal in both HS and GNS inner products for the reduced Gibbs state $\w{\si}_\beta = \tr_{\mc{A}^{\rm full}_{\rm e}}(\si_\beta) \propto \exp(- \beta \sum_p \bz_p)$. Moreover, we compute
\begin{align*}
    & e^{- \beta \sum_p \bz_p} A_1 = \eta^{-1}  A_1\,, \q e^{- \beta \sum_p \bz_p} A_2 = A_2\,, \\
    & e^{- \beta \sum_p \bz_p} A_3 = A_3\,, \q e^{- \beta \sum_p \bz_p} A_4 = \eta A_4\,,
\end{align*}
where $\eta = e^{- 2 \beta}$,
which implies that
\begin{align*}
    \|A_1\|_{\w{\sigma}_\beta} = \mc{Z}_\beta^{-1/2}\eta^{-1/2}\,,\quad \|A_2\|_{\w{\sigma}_\beta}=\|A_3\|_{\w{\sigma}_\beta} = \mc{Z}_\beta^{-1/2}\,,\quad \|A_4\|_{\w{\sigma}_\beta} = \mc{Z}_\beta^{-1/2}\eta^{1/2}\,,
\end{align*}
with $\mc{Z}_\beta$ being the partition function of $\sigma_\beta$.
    Then, we find, by using \eqref{2eq:local_lind},
\begin{align*}
    \mc{L}_{\sigma^x_j} (A_1)
    & = h_+ \ket{--}\bra{--} - h_- \ket{++}\bra{++} = h_+ A_4 - h_- A_1
\end{align*}
due to
\begin{align*}
    & [A_1, a_j] = [ \ket{++}\bra{++},   \ket{++}\bra{--}] = \ket{++}\bra{--}\,, \\
    & [A_1, a_j^\dag] = [ \ket{++}\bra{++},   \ket{--}\bra{++}] = - \ket{--}\bra{++}\,.
\end{align*}
Similarly, we have $\mc{L}_{\sigma^x_j} (A_4) = - h_+ A_4 + h_- A_1$, by
  \begin{align*}
        & [A_4, a_j] = [ \ket{--}\bra{--},   \ket{++}\bra{--}] = - \ket{++}\bra{--}\,, \\
        & [A_4, a_j^\dag] = [ \ket{--}\bra{--},   \ket{--}\bra{++}] = \ket{--}\bra{++}\,.
    \end{align*}
    In the same way, we can also compute
    \begin{equation*}
        \mc{L}_{\sigma^x_j}(A_2) = A_3-A_2\,,\q \mc{L}_{\sigma^x_j}(A_3) = A_2 - A_3\,.
    \end{equation*}
This allows us to compute the local matrix representation of the master Hamiltonian $\vp \circ \mc{L}_1\circ \vp^{-1}$ via $
\frac{\l A_i, \mc{L}_1 A_j \r_{\w{\sigma}_\beta}}{\norm{A_i}_{\w{\sigma}_\beta} \norm{A_j}_{\w{\sigma}_\beta}} $:
    \begin{equation} \label{matrix_ab}
        \left[\begin{matrix}
           - h_- = - \frac{2 \eta^2}{\eta^2 + 1} & 0 & 0 & h_-/\eta = \frac{2 \eta}{\eta^2 + 1} \\
0 & -1 & 1 & 0 \\
0 & 1 & -1 & 0   \\
\eta h_+ = \frac{2 \eta}{\eta^2 + 1} & 0 & 0 & -h_+ = - \frac{2}{\eta^2 + 1}\\
        \end{matrix}\right]\,.
\end{equation}

\item For $j \in  \Gamma_{\rm flip}$, we let
    \begin{align*}
        & A_1 = \ket{++}\bra{--}\,,\q A_2 = \ket{+ - }\bra{- +}\,,\q A_3 = \ket{- +}\bra{+ -}\,,\q  A_4 = \ket{--}\bra{++}\,.
    \end{align*}
        For $A_1,A_4$, noting that
    \begin{align*}
         & [A_1, a_j] = [ \ket{++}\bra{--},   \ket{++}\bra{--}] = 0\,, \\
        & [A_1, a_j^\dag] = [ \ket{++}\bra{--},   \ket{--}\bra{++}] = \ket{++}\bra{++} - \ket{--}\bra{--}\,,
    \end{align*}
    and
    \begin{align*}
          & [A_4, a_j] = [ \ket{--}\bra{++},   \ket{++}\bra{--}] =  \ket{--}\bra{--}- \ket{++}\bra{++}\,, \\
        & [A_4, a_j^\dag] = [ \ket{--}\bra{++},   \ket{--}\bra{++}] = 0\,,
    \end{align*}
    we have
    \begin{align*}
        \mc{L}_{\sigma^x_j} (A_1)   & = \frac{1}{2} \left\{h_+ [a_j^\dag ,A_1] a_j + h_- a_j [A_1, a_j^\dag] \right\}
        = -\frac{1}{2} (h_+ + h_-) A_1\,,
    \end{align*}
    and
    \begin{align*}
         &  \mc{L}_{\sigma^x_j} (A_4) = \frac{1}{2} \left\{ h_+ a_j^\dag [A_4, a_j] + h_- [a_j, A_4] a_j^\dag  \right\} = - \frac{1}{2} (h_+ + h_-) A_4\,.
    \end{align*}
Similarly, a direct computation also gives
    \begin{equation*}
        \mc{L}_{\sigma^x_j}(A_2) = A_3 - A_2\,,\q \mc{L}_{\sigma^x_j}(A_3) = A_2 - A_3\,.
    \end{equation*}

The local matrix representation of $\vp \circ \mc{L}_1\circ \vp^{-1}$ via $
\frac{\l A_i, \mc{L}_1 A_j \r_{\w{\sigma}_\beta}}{\norm{A_i}_{\w{\sigma}_\beta} \norm{A_j}_{\w{\sigma}_\beta}} $ is given by
 \begin{equation} \label{matrix_flip}
        \left[\begin{matrix}
           - \frac{h_+ + h_-}{2} & 0 & 0 & 0 \\
            0 & -1 & 1 & 0 \\
            0 & 1 & -1 & 0   \\
            0 & 0 & 0 & - \frac{h_+ + h_-}{2}\\
        \end{matrix}\right]\,.
    \end{equation}

\item For $j \in \Gamma_{\rm int}$, we let\footnote{Without loss of generality, we place $j \in \Lambda$ in the second position. The other case of $j-1\in \Lambda$ is symmetric.}
    \begin{align*}
        A_1 = \ket{++}\bra{-+}\,, \q  A_2 = \ket{+-}\bra{--}\,, \q  A_3 = \ket{-+}\bra{++}\,, \q  A_4 = \ket{--}\bra{+-}\,.
    \end{align*}
    and find that they are eigenvectors of $\mc{L}_{\sigma^x_j}$:
    \begin{equation*}
        \mc{L}_{\sigma^x_j}(A_1) =  - \frac{h_- + 1}{2} A_1\,,\q \mc{L}_{\sigma^x_j}\left(A_2\right) = - \frac{h_+ + 1}{2} A_2\,,
    \end{equation*}
    and
   \begin{equation*}
        \mc{L}_{\sigma^x_j}(A_3) =  - \frac{h_- + 1}{2} A_3,\quad \mc{L}_{\sigma^x_j}(A_4) =  - \frac{h_+ + 1}{2} A_4\,.
   \end{equation*}
   In this case, the local matrix representation of $\vp \circ \mc{L}_1\circ \vp^{-1}$ via $
\frac{\l A_i, \mc{L}_1 A_j \r_{\w{\sigma}_\beta}}{\norm{A_i}_{\w{\sigma}_\beta} \norm{A_j}_{\w{\sigma}_\beta}} $ is
       \begin{equation} \label{matrix_int}
        \left[\begin{matrix}
           - \frac{h_- + 1}{2} & 0 & 0 & 0 \\
0 & - \frac{h_+ + 1}{2} & 0 & 0 \\
0 & 0 & - \frac{h_- + 1}{2} & 0   \\
0 & 0 & 0 & - \frac{h_+ + 1}{2}\\
        \end{matrix}\right]\,.
\end{equation}
 \end{itemize}
 The above calculation concludes the proof that $\mc{L}_1$ is block diagonal for the decomposition of $\mc{A}^{\rm full}_{\rm m}$ in \eqref{eq:decom}, namely, $\mc{L}_1: \mc{B}(\mc{H}_+)(\Lambda) \to \mc{B}(\mc{H}_+)(\Lambda)$ for any $\Lambda \subset \{1,2,\dots,L^2\}$ such that $L^2-|\Lambda|$ is even.
\end{proof}

\begin{proof}[Proof of \cref{lem:gap_abelian}]
 From \eqref{matrix_ab}, the matrix representation of the master Hamiltonian of $- \mc{L}_1$ can be written as:
\begin{equation*}
    K_{L^2}^\beta = K_{L^2} + (K_{L^2}^\beta - K_{L^2})\,,
\end{equation*}
where $K_{L^2}$ is the matrix representation at zero temperature (i.e., $\beta \to \infty$):
\begin{equation} \label{eq:matrixrepl2}
    K_{L^2}=\underbrace{\sum^{L^2-2}_{i=0}I_i\otimes k_T\otimes I_{L^2-2-i}}_{:=K_T}+\underbrace{\sum^{L^2-2}_{i=0}I_i\otimes k_D\otimes I_{L^2-2-i}}_{:=K_D}\,,
\end{equation}
with
\begin{equation}\label{eqn:small_k_t_d}
k_T= \left[\begin{matrix}
0 & 0 & 0 & 0 \\
0 & 1 & -1 & 0 \\
0 & -1 & 1 & 0   \\
0 & 0 & 0 & 0\\
\end{matrix}\right],\quad
k_D = \left[\begin{matrix}
0 & 0 & 0 & 0 \\
0 & 0 & 0 & 0 \\
0 & 0 & 0  & 0   \\
0 & 0 & 0 &  2\\
        \end{matrix}\right]\,.
\end{equation}
and $K_{L^2}^\beta - K_{L^2}$ is given by
\begin{equation*}
    K_{L^2}^\beta - K_{L^2} = \sum^{N-2}_{i=0}I_i\otimes \left[\begin{matrix}
\frac{2 \eta^2}{\eta^2 + 1} & 0 & 0 & - \frac{2 \eta}{\eta^2 + 1} \\
0 & 0 & 0 & 0 \\
0 & 0 & 0  & 0   \\
- \frac{2 \eta}{\eta^2 + 1} & 0 & 0 &  - \frac{2 \eta^2}{\eta^2 + 1}\\
        \end{matrix}\right] \otimes I_{L^2-2-i}\,,
\end{equation*}
where $\eta=\exp(-2\beta)$. The operator norm of the self-adjoint operator $K_{L^2}^\beta - K_{L^2}$ can be directly estimated as
\begin{equation*}
    \norm{K_{L^2}^\beta - K_{L^2}} = \Theta (L^2 e^{-2 \beta}) = \Or(L^{-6})\,, \quad \text{when $\beta \ge 4 \ln L$\,.}
\end{equation*}
Therefore, to obtain the desired gap estimate \eqref{eqn:gap_wt_L}, we only need to consider the gap at zero temperature and prove $\gap(K_{L^2}) = \Omega(L^{-4})$.

\begin{figure}[bthp]
\centering
\includegraphics[width=0.8\textwidth]{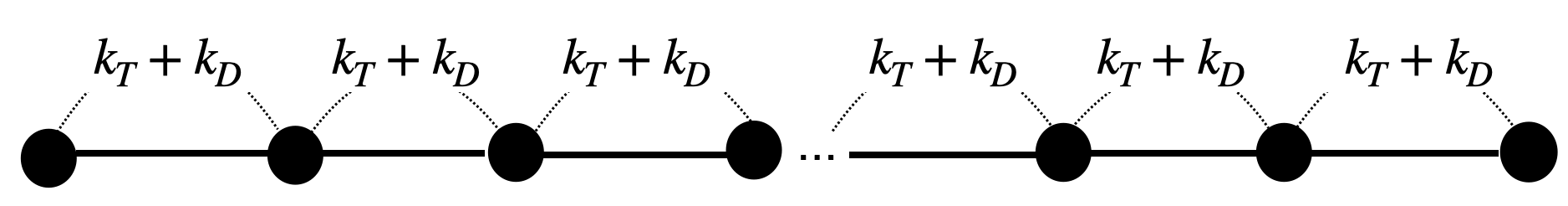}
\caption{Action of $\mc{L}_1$ ($\beta \to \infty$) on $\Gamma_{\rm ab}$. $k_T$ and $k_D$ are defined in \eqref{eqn:small_k_t_d}. Each black dot represents a bond corresponding to the plaquette observable $\bz_j$ ordered by the snake.}
\label{fig:1D_abelian}
\end{figure}

For this, we consider the following configuration space
\begin{equation} \label{eq:spacesn}
\mathcal{S}_{L^2}=\mathrm{Span}\left(\mathcal{A}_{L^2}\right) \cong  \C^{2^{L^2}},\quad \text{with $\mathcal{A}_{L^2} = \left\{+,-\right\}^{\otimes {L^2}}$}\,,
\end{equation}
with the space decomposition $\mathcal{S}_{L^2}=\bigoplus_k \mathcal{S}_{{L^2},k}$, where $\mathcal{S}_{{L^2},k}=\mathrm{Span}\left(\mathcal{A}_{{L^2},k}\right)$ with
\[
\mathcal{A}_{{L^2},k}:= \left\{a\in\left\{+,-\right\}^{\otimes {L^2}}~\middle|~\text{$a$ has $k$ $``-"$ signs}\right\}\,.
\]
It is clear from the construction that the action of $- \mc{L}_1$ on $\mc{A}^{\rm full}_{\rm m}(\Lambda)$ with $\Lambda = \Gamma_{\rm ab} = \{1,2,\ldots\}$ is the same as the action of $K^\beta_{L^2}$
on $\bigoplus_{k\,\text{even}} \mathcal{S}_{{L^2},k}$ under the identification $\ket{m}\bra{m}\rightarrow \ket{m}$. One can also readily check that
\begin{equation} \label{eq:invariant}
    K_{L^2}: \mathcal{S}_{L^2,k} \to \mathcal{S}_{L^2,k}\,,
\end{equation}
that is, $K_{L^2}$ is block diagonal for the decomposition $\mathcal{S}_{L^2}=\bigoplus_k \mathcal{S}_{L^2,k}$.


When $k=0$, we have $K_{L^2}|_{\mathcal{A}_{L^2,0}}=0$ and $\mathrm{dim}\left(\mc{A}_{L^2,0}\right)=1$. Thus,
\[
\mathrm{Gap}\left(-\mc{L}_1|_{\mc{A}^{\rm full}_{\rm m}([L^2])}\right)\geq \min_{k=2,3,4,\cdots}\lambda_{\min}\left(K_{L^2}|_{\mathcal{S}_{L^2,k}}\right)\,.
\]
In principle, we only need to consider the admissible configuration, that is, the subspaces $\mc{S}_{L^2,k}$ with even $k$. However, it is simpler to prove a result for all $k$'s via iterative reduction. We emphasize that such a relaxation does not produce any additional dependence on the system size $N = 2L^2$, and hence will not give a worse spectral gap lower bound.


We now consider the lower bound of $K_{L^2}$ on each $\mathcal{S}_{L^2,k}$ for $k>0$. Define
\[
\lambda_{L^2,k}:=\lambda_{\min}\left(K_{L^2}|_{\mathcal{S}_{L^2,k}}\right)\,.
\]

We first note that when $k=L^2$, it also holds that
$\dim\left(\mathcal{S}_{L^2,L^2}\right)=1$, and
\begin{equation*}
K_{L^2}|_{\mathcal{S}_{L^2,L^2}} = K_D|_{\mathcal{S}_{L^2,L^2}} \succeq 2 (L^2 - 1) \ge 1\,.
\end{equation*}

When $ k \geq 3 $ and $ L^2 \geq 3 $,  we use the following iteration to reduce the estimation of $\lad_{L^2,k}$ to $\lad_{L^2, 2}$. By the representation \eqref{eq:matrixrepl2} of $K_{L^2}$, we find
\begin{equation*}
     K_{L^2}=K_{L^2-1}\otimes I_1+I_{L^2-2}\otimes (k_T+k_D)\,.
\end{equation*}
In addition, there holds
\[
\mathcal{S}_{L^2,k}=\mathrm{Span}\left(\mathcal{A}_{L^2-1,k}\otimes \ket{+}\right)\oplus \mathrm{Span}\left(\mathcal{A}_{L^2-1,k-1}\otimes \ket{-}\right)\,.
\]
Thus, for any given
 $a=a_+\otimes \ket{+}+a_-\otimes \ket{-}\in \mathcal{S}_{L^2,k}$ with $|a_+|^2+|a_-|^2=1$, we can derive
 \begin{equation} \label{eqn:iteration}
     \begin{aligned}
          a^* K_{L^2} a &\geq a^* (K_{L^2-1}\otimes I_1) a \\ & =(a_+)^* K_{L^2-1}a_++(a_-)^* K_{L^2-1}a_- \\ &\geq \min\left\{\lad_{L^2-1,k}, \lad_{L^2-1,k-1}\right\}\,,
     \end{aligned}
 \end{equation}
where the first inequality follows from $k_T + k_D\succeq 0$, the second inequality follows from \eqref{eq:invariant} and $a_+\in \mc{S}_{L^2-1,k}$, $a_-\in \mc{S}_{L^2-1, k-1}$.

Therefore, it suffices to estimate $\lad_{L^2,2}$ to finish. To do so, we find, by \eqref{eq:matrixrepl2},
\begin{equation*}
K_{L^2}|_{\mathcal{S}_{L^2,2}}=(K_T)|_{\mathcal{S}_{L^2,2}}+(K_D)|_{\mathcal{S}_{L^2,2}}\,.
\end{equation*}

Note that the basis of the subspace ${\cal S}_{L^2, 2}$ can be denoted by:
\begin{align*}
    \left\{|i,j\rangle:1\leq i\leq L^2 - 1, i+1\leq j\leq L^2\right\}\,.
\end{align*}
This represents the positions of the two ``$-$" signs and defines a stair graph $G_{L^2-1}$ with parameter $L^2-1$ (see \cref{def:stair_graph} and \cref{fig:stair_graph} in \cref{sec:spectral_graph}).  From the representation \eqref{eq:matrixrepl2} and \eqref{eqn:small_k_t_d}, there holds
\begin{align*}
    K_T |i,j\rangle = \textbf{1}_{i>1} (|i,j\rangle-|i-1,j\rangle) + \textbf{1}_{i<j-1}(|i,j\rangle-|i+1, j\rangle) \\
    + \textbf{1}_{j>i+1}(|i,j\rangle -|i, j-1\rangle) + \textbf{1}_{j<L^2}(|i,j\rangle -|i,j+1\rangle)\,,
\end{align*}
where the first and third term represent the transition from $\ket{+-}$ to $\ket{-+}$ and the second and fourth term represent the transition from $\ket{-+}$ to $\ket{+-}$. It is straightforward to see that $(K_T)|_{{\cal S}_{L^2, 2}}$ is exactly the graph Laplacian associated with the stair graph $G_{L^2-1}$.  Moreover, $(K_D)|_{{\cal S}_{L^2, 2}}$ can be identified with the diagonal matrix $D_{L^2-1}$ in \cref{thm:graph_spectral}:
\begin{align*}
    K_D |i,j\rangle = \textbf{1}_{j=i+1}\cdot 2|i,i+1\rangle\,.
\end{align*}
Hence, by \cref{thm:graph_spectral} with $n = L^2  - 1$,
\begin{align}\label{eq:estladl2}
    \lambda_{L^2, 2} = \lambda_{\min}((K_T)|_{\mathcal{S}_{L^2,2}}+(K_D)|_{\mathcal{S}_{L^2,2}})= \lambda_{\min}(H_{L^2-1}) = \Omega\left(L^{-4}\right)\,.
\end{align}
Using \cref{eqn:iteration}, it follows that $\min_{3\leq k\leq L^2}\lambda_{L^2,k}=\Omega\left(L^{-4}\right)$.
Therefore, we have proved
\begin{align*}   \mathrm{Gap}\left(K_{L^2}\right)=\Omega\left(L^{-4}\right) = \Omega\left(N^{-2}\right).
\end{align*}
\end{proof}

\begin{proof}[Proof of \cref{lem:gap_abelian_comb}]
Following the notation in the proof of \cref{lem:gap_abelian}, we still consider the subspace
$\mathcal{S}_{L^2}$ in \eqref{eq:spacesn} and denote by $K^\beta_{L^2}$ the matrix
representation of the master Hamiltonian of $- \mc{L}_2$. Moreover, we similarly have
\begin{equation*}
    \norm{K_{L^2}^\beta - K_{L^2}} = \Theta (L^2 e^{-2 \beta}) = \Or(L^{-6})\,,\q \text{for $\beta\geq 4 \ln L$}\,,
\end{equation*}
where $K_{L^2} := K^\infty_{L^2}$. One can also see that each local term in $K_{L^2}$ has the same form as the one in \eqref{eqn:small_k_t_d}, but the tensor structure is different\footnote{Since the qubits and star observables on the comb cannot be ordered along a line, some local term in $K_{L^2}$ is of the form $I \otimes a \otimes I \otimes b \otimes I$ with $a, b$ being non-identity $2 \times 2$ matrix. Since some observable $\bx_s$ is altered by three $\mc{L}_{\si_j^z}$, there are some sites where we find three local terms in $K_{L^2}$ nontrivially acting on it.}; see \cref{fig:2d1}.  We then decompose $\mathcal{S}_{L^2}$ according to number of $``-"$ signs in the entry of basis:
\begin{equation}\label{eqn:S_decomposition}
\mathcal{S}_{L^2}=\bigoplus_{k} \mathcal{S}_{L^2,k}\,,\q \mathcal{S}_{{L^2},k}=\mathrm{Span}\left(\mathcal{A}_{{L^2},k}\right)\,,
\end{equation}
with
\[
\mathcal{A}_{{L^2},k}:=\left\{a\in\left\{+,-\right\}^{\otimes {L^2}}\middle|\text{$a$ has $k$ $``-"$ signs}\right\}\,.
\]
Then, $K_{L^2}$ is block diagonal for \eqref{eqn:S_decomposition} and
\[
\mathrm{Gap}\left(K_{L^2}\right)\geq \min_{k=2,3,4,\cdots}\lambda_{\min}\left(K_{L^2}|_{\mathcal{S}_{{L^2},k}}\right)\,.
\]


We first consider the subspace $\mathcal{S}_{{L^2},2}$, whose basis vectors contain only two $``-"$ signs. Our strategy for lower bounding $\lambda_{\min}\left(K_{L^2}|_{\mathcal{S}_{{L^2},2}}\right)$ is to reduce this problem to a straight line case as in \cref{fig:1D_abelian}. For this, we introduce a set of lines covering all vertices in the comb. More specifically, we may regard the comb as a connected graph (more precisely, a tree) with $L^2$ vertices of degree at most 3, each corresponding to a star $\bx_s$ that interacts with the comb. Let ${\bf D}_{\rm comb}$ be the set of degree-one vertices (i.e., end-points) in the comb, and $l_{u,v}$ be the shortest path between the vertices $u$ and $v$ in the comb. Define
\begin{align}\label{eqn:line_E}
    {\bf P}_{\rm comb}:=\{l_{u,v}:\forall~u\ne v\in {\bf D}_{\rm comb}\}\,.
\end{align}
Then, we know that ${\bf P}_{\rm comb}$ contains $L(L-1)/2 = \mc{O}(L^2)$ (simple) paths,
and the maximum length of the paths is $\ell_{\rm comb} = 3L-2$.
Two examples of these paths are given in \cref{fig:2d2}.

For an arbitrary unit vector $\alpha\in \mathcal{S}_{{L^2},2}$:
\begin{align*}
    \alpha=\sum_{a\in \mc{A}_{L^2,2}}p_a\ket{a}\,,
\end{align*}
by \cref{fac:comb_line_cover}, there exists a path $\tilde{l}$ in ${\bf P_{\rm comb}}$ such that for
\begin{align*}
\alpha_{\tilde{l}}:=\sum_{a\in \mc{A}_{L^2,k}:\#\{s\in \tilde{l}:a_s=``-"\}>1} p_a \ket{a}\,,
\end{align*}
it holds that
\begin{align}\label{eq:alpha_l*_2_norm}
    \|\alpha_{\tilde{l}}\|_2^2=\Omega(1/L^2)\,.
\end{align}
Let $(K_{L^2})_{\tilde{l}}$ be the restriction of $K_{L^2}$ to the path $\tilde{l}$. More specifically,
\begin{align*}
    (K_{L^2})_{\tilde{l}} :=\sum_{e=(u,v)\in \tilde{l}} (k_T+k_D)_{u,v}\,,
\end{align*}
where $(k_T+k_D)_{u,v}$ is a local term that applies $k_T+k_D$ to the ``qubits'' at $u$ and $v$.
Then, we have
\begin{equation} \label{eq:est1}
\begin{aligned}
    \alpha^* K_{L^2} \alpha \geq &~  \alpha^* (K_{L^2})_{\tilde{l}} \alpha\\
    \geq &~ \alpha^*_{\tilde{l}} (K_{L^2})_{\tilde{l}} \alpha_{\tilde{l}}\\
    \geq &~ \|\alpha_{\tilde{l}}\|_2^2\cdot \lambda_{\ell_{\rm comb}, 2}\\
    = &~ \mmg{\Omega(L^{-2})\cdot \Omega(L^{-2})=\Omega(L^{-4})}\,,
\end{aligned}
\end{equation}
where the first step follows from $K_{L^2}-(K_{L^2})_{\tilde{l}}$ is positive semi-definite, the second step follows from $\langle \alpha-\alpha_{\tilde{l}}, (K_{L^2})_{\tilde{l}} \alpha_{\tilde{l}}\rangle = 0$, the third step follows from $(K_{L^2})_{\tilde{l}}$ is equivalent to a 1D chain of length $|\tilde{l}|\leq \ell_{\rm comb}=\Or(L)$ and $a_{\tilde{l}}\in \mc{S}_{|\tilde{l}|,2}$, and the last step follows from \eqref{eq:alpha_l*_2_norm} and \eqref{eq:estladl2}.
Thus, we conclude
\begin{equation}\label{eq:comb_est}
\mmg{K_{L^2}|_{\mathcal{S}_{{L^2},2}}=\Omega(L^{-4})=\Omega(N^{-2})}\,.
\end{equation}

We proceed to estimate $\lambda_{\min}\big(K_{L^2}|_{\mathcal{S}_{{L^2},k}}\big)$ for $k\geq 3$. For an arbitrary unit vector $\alpha\in \mathcal{S}_{{L^2},k}$:
\begin{align*}
    \alpha=\sum_{a\in \mc{A}_{L^2,k}} p_a \ket{a}\,,
\end{align*}
by \cref{fac:comb_line_cover}, there exists a path $\tilde{l}$ in ${\bf P_{\rm comb}}$ such that for
\begin{align*}
    \alpha_{\tilde{l},q}:=\sum_{\substack{a\in \mc{A}_{L^2,k}:\\\#\{s\in \tilde{l}:a_s=``-"\}=q}} p_a \ket{a}~~~\forall 2\leq q\leq k\,,
\end{align*}
we have
\begin{align}\label{eq:alpha_l*_q_length}
    \sum_{q=2}^k \|\alpha_{\tilde{l},q}\|_2^2 = \Omega(1/L^2)\,.
\end{align}
Let $(K_{L^2})_{\tilde{l}}$ be the restriction of $K_{L^2}$ to the path $\tilde{l}$.  Then, we find, similar to \eqref{eq:est1},
\begin{align*}
\alpha^* K_{L^2} \alpha\geq &~ \alpha^* (K_{L^2})_{\tilde{l}} \alpha\\
\geq &~ \sum^k_{q=2} \alpha^*_{\tilde{l},q} (K_{L^2})_{\tilde{l}} \alpha_{\tilde{l},q}\\
\geq &~ \sum_{q=2}^k \|\alpha_{\tilde{l},q}\|_2^2\cdot \lambda_{\ell_{\rm comb},q}\\
\geq &~  \sum_{q=2}^k \|\alpha_{\tilde{l},q}\|_2^2 \cdot \lambda_{\ell_{\rm comb}, 2}\\
= &~ \mmg{\Omega(L^{-2})\cdot \Omega(L^{-2}) = \Omega(L^{-4})}\,,
\end{align*}
where the first step follows from $K_{L^2}-(K_{L^2})_{\tilde{l}}$ is positive semi-definite, the second step follows from $\alpha_{\tilde{l},q_1}^* (K_{L^2})_{\tilde{l}} \alpha_{\tilde{l},q_2}=0$ for $q_1\neq q_2$, the third step follows from $(K_{L^2})_{\tilde{l}}$ is equivalent to a 1D chain of length $|\tilde{l}|\leq \ell_{\rm comb}=\Or(L)$, and $\alpha_{\tilde{l},q}$ is in the subspace $\mc{S}_{|\tilde{l}|, q}$, the fourth step follows from the recursion relation \eqref{eqn:iteration} that $\lambda_{\ell_{\rm comb},q}\geq \lambda_{\ell_{\rm comb},2}$, an the last step follow from \eqref{eq:estladl2} and \eqref{eq:alpha_l*_q_length}. 
Thus, we have
\begin{align*}
K_{L^2}|_{\mathcal{S}_{{L^2},k}}=\mmg{\Omega(L^{-4})=\Omega(N^{-2})}\,.
\end{align*}

Combining them all together, we now can conclude
\[
\begin{aligned}
\mathrm{Gap}\left(-\mathcal{L}_2|_{\mc{A}_{\rm e}^{\rm full}([L^2])}\right) \ge \min_{2\le k\leq L^2}\lambda_{\min}\left(K_{L^2}|_{\mathcal{S}_{{L^2},k}}\right) \mmg{- \Or(N^{-3}) = \Omega(N^{-2})}\,.
\end{aligned}
\]
This concludes the proof of \cref{lem:gap_abelian_comb}.
\end{proof}

\begin{fact}\label{fac:comb_line_cover}
Let $2\leq k \leq L^2-1$. For any unit vector $\alpha\in \mc{S}_{L^2,k}$ of the form:
\begin{align*}
    \alpha = \sum_{a\in \mc{A}_{L^2,k}} p_a \ket{a}\,,
\end{align*}
there exists a path $\tilde{l}\in {\bf P_{\rm comb}}$ defined as \eqref{eqn:line_E} such that
\begin{align*}
    \sum_{\substack{a\in \mc{A}_{L^2,k}:\\\#\{s\in \tilde{l}:a_s=``-"\}>1}}|p_a|^2=\mmg{\Omega(1/L^2)}\,.
\end{align*}
\end{fact}
\begin{proof}
We first observe that for any $a\in \{+,-\}^{L^2}$ with $k$ $``-"$, by the construction of ${\bf P}_{\rm comb}$, there exists at least one path $l\in {\bf P}_{\rm comb}$ that contains at last two $``-"$ on it.
Then, we have
\begin{align*}
    \sum_{l\in {\bf P}_{\rm comb}}\sum_{\substack{a\in \mc{A}_{L^2, k}:\\\#\{s\in l:a_s=``-"\}>1}} |p_a|^2=&~ \sum_{a\in \mc{A}_{L^2,k}} |p_a|^2\cdot \#\{l\in {\bf P}_{\rm comb}:\#\{s\in l:a_s=``-"\}>1\}\\
    \geq &~ \sum_{a\in \mc{A}_{L^2,k}} |p_a|^2\cdot 1\\
    = &~ 1\,,
\end{align*}
where the first step follows from exchanging the summations, the second step follows from our observation, and the last step follows from $\alpha$ is a unit vector.
Since ${\bf P}_{\rm comb}$ contains $L^2$ paths, by the averaging argument, there must exists an $\tilde{l}\in {\bf P}_{\rm comb}$ such that
\begin{align*}
    \sum_{\substack{a\in \mc{A}_{L^2, k}:\\\#\{s\in l:a_s=``-"\}>1}} |p_a|^2 \geq \mmg{1/L^2}\,,
\end{align*}
which proves the proposition.
\end{proof}

\noindent \textbf{Data availability}. Data sharing not applicable to this article as no datasets were generated or analyzed during
the current study.

\noindent \textbf{Conflict of interest}. The authors have no conflicts of interest to declare that are relevant to the content of this article.

\bibliographystyle{alpha}
\bibliography{ref}

\appendix

\appendix

\section{Spectral graph theory problem}\label{sec:spectral_graph}

In this appendix, we study the spectral property of the stair graph, which is the grid graph folded along the diagonal.
Let $n$ be the side length of the stair. We consider the graph Laplacian of the stair graph of size $\binom{n+1}{2}$,
perturbed by adding positive weights to the vertex on the diagonal.
We shall prove that the minimum eigenvalue of such class of perturbed Laplacians is $\Theta(n^{-2})$. The result is of independent interest and is applied to the proofs of \cref{lem:gap_abelian,lem:gap_abelian_comb}, see~\eqref{eq:estladl2} and~\eqref{eq:comb_est} for detail.

We first introduce the definition of the stair graph, graph Laplacian, and their basic properties.


\begin{defn}[Stair graph]\label{def:stair_graph}
A stair graph $G=(V,E)$ with parameter $n\in \mathbb{N}_+$ is an undirected graph with $\binom{n+1}{2}=\frac{n(n+1)}{2}$ vertices. It can be viewed as an $n$-by-$n$ grid folded along the diagonal (see \cref{fig:stair_graph} for an example).

More specifically, let the rows be labeled by $1,2,\dots,n$ (from bottom to top), and the columns be labeled by $2,3,\dots,n+1$ (from right to left)\footnote{These labels are chosen to align with the notation used in the proof of~\cref{lem:gap_abelian,lem:gap_abelian_comb}.}. Then, each vertex gets a label $(i,j)$ based on its row and column indices.

For a vertex $(i,j)$, it is connected to its nearest-neighbors $(i,j\pm 1)$ and $(i\pm 1, j)$, with open boundary condition.
\end{defn}

\begin{figure}[!ht]
    \centering
    \includegraphics[width=0.8\linewidth]{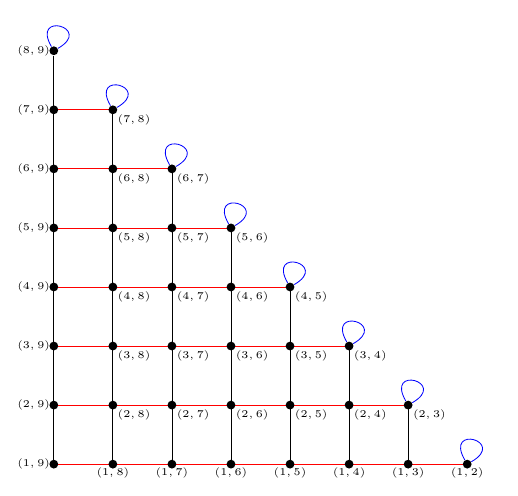}
    \caption{An example of the stair graph ($n=8$) with the vertex labels. The red edges are the horizontal edges $E_{8,h}$.  The blue self-loops are the entries in $D_8$ of value $2$.}
    \label{fig:stair_graph}
\end{figure}

\begin{defn}[Graph Laplacian]
For an undirected graph $G=(V,E)$, its graph Laplacian $\mathcal{L}_G$ is defined to be:
\begin{align*}
    {\cal L}_G:=D_G - A_G\,,
\end{align*}
where $D_G=\diag(\{\deg(v)\}_{v\in V})$ and $A_G$ is the adjacency matrix of $G$.
\end{defn}

\begin{prop}[Spectral properties of graph Laplacian]
For any undirected graph $G$, its Laplacian ${\cal L}_G$ satisfies the following properties:
\begin{itemize}
    \item ${\cal L}_G$ is a real-symmetric, positive semi-definite matrix with minimum eigenvalue $0$.
    \item The number of zero eigenvalues is equal to the number of connected components of $G$.
\end{itemize}
\end{prop}

Now, we are ready to present the main theorem in this section:

\begin{thm}\label{thm:graph_spectral}
For any $n\in \mathbb{N}_+$, let $G_n=(V_n,E_n)$ be the stair graph with parameter $n$, and ${\cal L}_n$ be its graph Laplacian.
Let $D_n:V_n\times V_n\rightarrow \mathbb{R}$ be a diagonal matrix such that $D_n((i,i+1),(i,i+1))=2$ and other entries equal to 0. Then
\begin{align*}
    H_n:={\cal L}_n + D_n\succeq \Omega\left(\frac{1}{n^2}\right)\,.
\end{align*}
\end{thm}

\begin{proof}
For any vector $f\in\mathbb{R}^{V_n}$, we have
\begin{align*}
    f^\top {\cal L}_n f = \sum_{(u,v)\in E_n} (f(u)-f(v))^2\,.
\end{align*}
Thus,
\begin{align*}
    f^\top H_n f = \sum_{(u,v)\in E_n} (f(u)-f(v))^2 + \sum_{i=1}^n 2f(i,i+1)^2\,.
\end{align*}
In the second term, we abuse the notation and use $f(i,j)$ to denote $f(v)$ with $v=(i,j)$ for simplicity.

For any off-diagonal vertex $v=(i,j)$ with $i+2\leq j\leq n+1$, consider the horizontal path in the stair graph:
\begin{align*}
    (i,j)\rightarrow (i,j-1)\rightarrow (i,j-2)\rightarrow \cdots \rightarrow (i,i+1)\,.
\end{align*}
We can apply telescoping summation and Cauchy-Schwartz inequality along this path:
\begin{align*}
    f(i,j) = f(i,i+1) + \sum_{k=i+1}^{j-1} (f(i,k+1) - f(i,k))\,,
\end{align*}
which implies that
\begin{align*}
    f(i,j)^2 \leq &~ 2f(i,i+1)^2 + \REV{2} \left(\sum_{k=i+1}^{j-1} (f(i,k+1) - f(i,k))\right)^2\tag{$(a+b)^2\leq 2a^2 + 2b^2$}\\
    \leq &~ 2f(i,i+1)^2 + \REV{2} (j-i-1)\cdot \sum_{k=i+1}^{j-1} (f(i,k+1) - f(i,k))^2\,.\tag{Cauchy-Schwartz}
\end{align*}

Next, we sum over all vertices:
\begin{align*}
    \sum_{v\in V_n} f(v)^2=&~\sum_{i=1}^n f(i,i+1)^2 + \sum_{i=1}^{n-1} \sum_{j=i+2}^{n+1} f(i,j)^2\\
    \leq &~ \sum_{i=1}^n f(i,i+1)^2 + \sum_{i=1}^{n-1} \sum_{j=i+2}^{n+1}\left(2f(i,i+1)^2 + \REV{2} (j-i-1)\cdot \sum_{k=i+1}^{j-1} (f(i,k+1) - f(i,k))^2\right)\\
    =&~ \sum_{i=1}^n (2(n-i)+1) f(i,i+1)^2 + \REV{2} \sum_{i=1}^{n-1}\sum_{k=i+1}^{n}(f(i,k+1) - f(i,k))^2\sum_{j=k+1}^{n+1}(j-i-1)\\
    = &~ \sum_{i=1}^n (2(n-i)+1) f(i,i+1)^2 + \sum_{i=1}^{n-1}\sum_{k=i+1}^{n}(f(i,k+1) - f(i,k))^2 \cdot \REV{(n+k-2i)(n-k+1)}\\
    \leq &~ \Or(n)\cdot \sum_{i=1}^n 2f(i,i+1)^2 + \Or(n^2)\cdot \sum_{(u,v)\in E_{n,h}} (f(u)-f(v))^2\,,
\end{align*}
where $E_{n,h}$ are the set of horizontal edges in $G_n$, and the last step follows from $2(n-i)+1\leq 2n$ and $(n+k-2i)(n-k+1)\leq n(n-1)$.

Since $E_{n,h}\subset E_n$, we have
\begin{align*}
    \|f\|_2^2 = \sum_{v\in V_n}f(v)^2 \leq &~ \Or(n^2)\cdot \left(\sum_{i=1}^n 2f(i,i+1)^2 + \sum_{(u,v)\in E_{n}} (f(u)-f(v))^2\right)\\
    = &~ \Or(n^2)\cdot f^\top H_n f\,.
\end{align*}
By the variational characterization of eigenvalues,
\begin{align*}
    \lambda_{\min}(H_n)=\min_{f\in \mathbb{R}^{V_n}} \frac{f^\top H_n f}{\|f\|_2^2}= \Omega\left(\frac{1}{n^2}\right)\,.
\end{align*}
\end{proof}

The following lemma shows that the $1/{n^2}$ lower bound is tight:
\begin{lemma}\label{lem:lower_bound}
For $n\in \mathbb{N}_+$, let $H_n$ be defined as in \cref{thm:graph_spectral}. Then,
\begin{align*}
    \lambda_{\min}(H_n)\leq \Or\left(\frac{1}{n^2}\right)\,.
\end{align*}
\end{lemma}
\begin{proof}
We can construct a test vector $g\in\R^{V_n}$  such that:
\begin{align*}
    g(i,j):=j-i-1~~~\forall 1\leq i\leq n,\, i+1\leq j\leq n+1\,.
\end{align*}

Then, $g(i,i+1)=0$ for all $i\in [n]$. And for each edge $(u,v)\in E_n$, it is easy to check that
\begin{align*}
    (g(u) - g(v))^2 = 1\,.
\end{align*}
Thus,
\begin{align*}
    g^\top H_n g = 0 + \sum_{(u,v)\in E_{n}} (g(u) - g(v))^2 = |E_n|= n(n-1)\,.
\end{align*}

On the other hand,
\begin{align*}
    \|g\|_2^2 = \sum_{i=1}^n \sum_{j = i+1}^{n+1} (j-i-1)^2=\sum_{i=1}^n \sum_{k=1}^{n-i} k^2=\frac{1}{12}n^2(n^2-1)\,.
\end{align*}

Thus, by Rayleigh quotient,
\begin{align*}
    \lambda_{\min}(H_n) \leq \frac{g^\top H_n g}{\|g\|_2^2} = \Or\left(\frac{1}{n^2}\right)\,.
\end{align*}
\end{proof}

Therefore, we conclude that $\lambda_{\min}(H_n)=\Theta(n^{-2})$.

\begin{rem}
\cref{thm:graph_spectral} and \cref{lem:lower_bound} can be easily generalized to the setting where
\begin{align*}
    H_n(a):=\mathcal{L}_n + aD_n\,,
\end{align*}
where $a=\Omega(1/n)$. In this case, the proof of \cref{thm:graph_spectral} still implies that the minimum eigenvalue $\lambda_{\min}(H_n(a))=\Omega(1/n^2)$, and the upper bound construction in \cref{lem:lower_bound} does not depend on $a$. Therefore, we have $\lambda_{\min}(H_n(a))=\Theta(1/n^2)$.
\end{rem}

\section{Low temperature Gibbs state preparation for 1D ferromagnetic Ising model} \label{sec:1d_ising}


This section studies the low-temperature thermal state preparation for a 1D ferromagnetic Ising chain with the periodic boundary condition. Although the Hamiltonian of this model is classical (a diagonal matrix in the computational basis), the jump operators are quantum and involve a significant amount of non-diagonal elements. The analysis of this model also shows the generality of core techniques for the 2D toric code in \cref{sec:2D_toric}.

We consider $N$ spins on a ring structure, modeled by the Hilbert space $\mc{H} \cong \C^{2^N}$. The Hamiltonian is
\begin{equation} \label{model:ising}
    H^{\rm Ising} = - J \sum_{j = 1}^N  \sigma^z_j \sigma^z_{j + 1}\,,\q J > 0\,,
\end{equation}
with $\si_{N+1}^z = \si_1^z$ (see \cref{fig:1D_Ising}, top), where $\{\bz_j := \sigma^z_j \sigma^z_{j + 1}\}_{j = 1}^N$ are bond observables satisfying
\begin{equation} \label{eq:bond}
    \prod_{j=1}^N \bz_j = I\,.
\end{equation}
We propose the following Gibbs smaller for the above Ising model:
\begin{equation} \label{sampler:ising}
    \mc{L}_\beta = \sum_{j = 1}^N \left(\mc{L}_{\sigma^x_j} + \mc{L}_{\sigma^y_j} + \mc{L}_{\sigma^z_j}\right) + \mc{L}_{\xx}\,,
\end{equation}
which is the sum of local Lindbladians with Pauli couplings plus a global one with coupling operator $\xx:= \prod_{j = 1}^N \si_j^x$. Here $\mc{L}_{\si_j^{x/y/z}}$ and $\mc{L}_\xx$ are defined via \eqref{eqq:davies2}. Their explicit formulas can be similarly derived as in \eqref{2eq:local_lind}, \eqref{2eq:local_lindlz}, and \eqref{2eq:lx}. The main result of this section is the following spectral gap lower bound of $\mc{L}_\beta$ in \cref{sampler:ising}.

\begin{thm} \label{thm:ising_gap}
For the Davies generator \eqref{sampler:ising}, we have
\begin{equation*}
    \gap(\mc{L}_\beta) \ge \max\left\{\Theta(e^{- 4 \beta J}), \mmg{\Theta(N^{-2})}\right\}\,.
\end{equation*}
\end{thm}

\begin{figure}[bthp]
\centering
\includegraphics[width=0.8\textwidth]{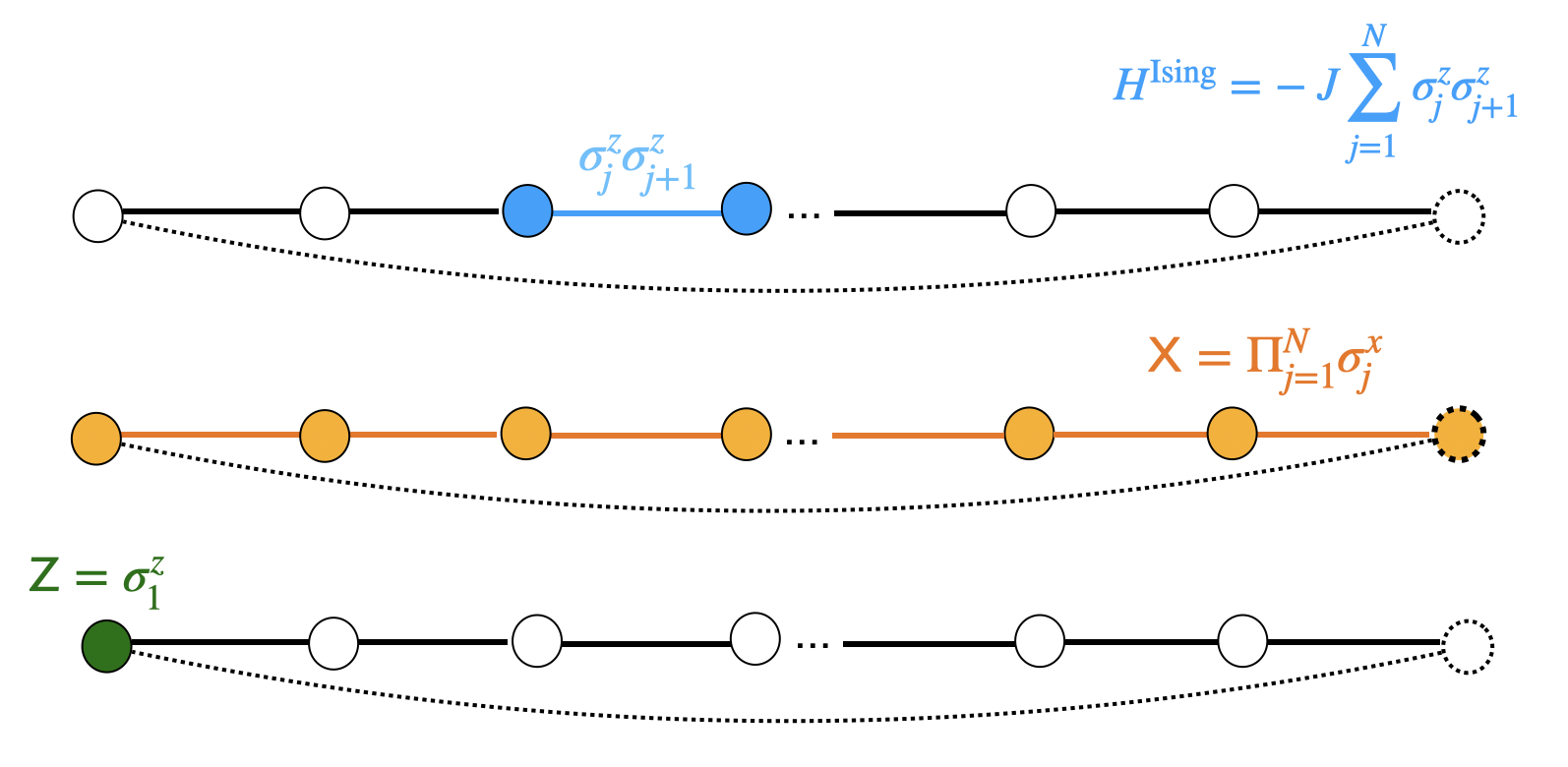}
\caption{1D Ising model. \emph{Top}: The 1D ferromagnetic Ising Hamiltonian $H^{\rm Ising}$. The dased dot and line denote the periodic boundary condition. \emph{Middle}: The logic operator $\xx = \Pi_{i = 1}^N \sigma^x_i$. \emph{Bottom}: The logic operator $\zz = \sigma^z_1$.}
\label{fig:1D_Ising}
\end{figure}

The proof is similar to that of \cref{thm:fast_mixing_2D}, based on the structures of the ground states of the Ising model and the associated observable algebra $\bh$. We know that $H^{\rm Ising}$ has a two-fold degenerate ground state
space spanned by
$\ket{0^N}$ and $\ket{1^N}$, and it is frustration-free, namely, any ground state $\ket{\vp}$ of  $H^{\rm Ising}$ satisfies $\sigma^z_j \sigma^z_{j + 1} \ket{\vp} = \ket{\vp}$ for each $j$.
This two two-fold degeneracy encodes a single logical qubit in the sense that with some abuse of notation, $\ket{0^N}$ and $\ket{1^N}$ can be identified as the logical qubit $\ket{0}$ and as $\ket{1}$, respectively:
\begin{equation*}
    \ket{0} \simeq \ket{0^N}\,,\q \ket{1} \simeq \ket{1^N}\,.
\end{equation*}
The excited states are given by acting Pauli string $X_{j_1} \cdots X_{j_n}$ on the ground states $\ket{0^N}$ and $\ket{1^N}$.
We define observables (see \cref{fig:1D_Ising})
\begin{align*}
    \xx = \prod_{i = 1}^N \sigma^x_i\,,\q  \zz = \sigma^z_1\,.
\end{align*}
The observable $\xx$ flips the logical qubit $\xx \ket{0/1}=\ket{1/0}$. The measurement outcome by $\zz$ indicates which state we are observing:
\begin{equation*}
    \bra{0}\zz\ket{0} = 1\,,\q  \bra{1}\zz\ket{1} = - 1\,.
\end{equation*}
We note that $\xx$ and $\zz$ commute with local terms $\{\sigma^z_j \sigma^z_{j + 1}\}$, and thus all the eigenspaces of $H^{\rm Ising}$ are the invariant subspaces of $\xx$ and $\zz$.


Let $\mc{Q} \cong \mc{B}(\C^2)$ be the algebra on the logical qubit generated by $\xx$ and $\zz$. Then we have $\mc{Q} = {\rm Span}\{I, \xx,\yy, \zz\}$, where $\yy = i \zz \xx$.
To describe the excited states and observable algebra $\bh$, we first define an admissible bond on a ring of $N$ sites as a configuration containing an even number of $-1$:
\begin{equation*}
    \ket{b} = \ket{b_1,b_2,\ldots, b_N}\in \left\{+1,-1\right\}^{N}\ \text{such that}\quad \#\left\{b_i=-1\right\} \in 2\ZZ\,.
\end{equation*}
The space spanned by the admissible bond, denoted by $\mc{H}_+$, is of dimension $2^{N-1}$.
An important observation is that
a natural orthonormal tensor basis of $\C^{2^N}$ consisting of $\ket{\epsilon_1\ldots\epsilon_N}$ with $\eps_i = 0/1$ can be uniquely written as $\ket{\eps_1}\ket{b}\in \C^2 \otimes \mc{H}_+$, where the first qubit
$\ket{\eps_1}$ is regarded as a logical qubit and $b_j$ is determined by the bond observable $\bz_j = \sigma^z_j \sigma^z_{j + 1}$:
\begin{equation*}
     \bz_j \ket{\epsilon_1\ldots\epsilon_N} = b_j \ket{\epsilon_1\ldots\epsilon_N}\,.
\end{equation*}
We define the full bond algebra $\mc{A}^{\rm full}_{b}$ by all the linear transformations on $\mc{H}_+$.
Then, as a consequence of the above decomposition of $\ket{\epsilon_1\ldots\epsilon_N}$, the observable algebra can be decomposed as follows~\cite[Lemma 4]{Alicki_2009}.
\begin{lemma}\label{lem:decom_1D_Ising}
The algebra of observables on a ring can be decomposed into
\begin{equation*}
    \mc{B}(\mc{H}) = \mc{Q} \otimes \mc{A}^{\rm full}_b\,.
\end{equation*}
In particular, we have the orthogonal decomposition in GNS inner product
\begin{equation} \label{eq:subspace_decom}
  \mc{B}(\mc{H}) = (I \otimes \mc{A}^{\rm full}_b)\oplus  (\xx \otimes \mc{A}^{\rm full}_b)\oplus  (\yy \otimes \mc{A}^{\rm full}_b)\oplus  (\zz \otimes \mc{A}^{\rm full}_b)\,.
\end{equation}
\end{lemma}

We are now ready to sketch the proof of \cref{thm:ising_gap}. For simplicity, we use $b_j = \pm$ for the bond configuration. Similarly to \cref{lem:2D_decomposition}, the algebra $\mc{A}_{b}^{\rm full}$ is generated by the observables $\{\bz_j\}^N_{j=1}$ and $\{\sigma^x_j\}_{j = 2}^N$, which commute with $\xx$. Then, by a direct computation, we have $\mathcal{L}_{\xx}\left(I\otimes \mathcal{A}^{\rm full}_{\rm b}\right)=0$. Moreover, there holds
\begin{equation} \label{eq:lxeig}
\mc{L}_\xx \left(\zz/\yy\otimes \mathcal{A}^{\rm full}_{\rm b}\right) = - 2 (\zz/\yy)\otimes \mathcal{A}^{\rm full}_{\rm b}, \q  - \l (\zz/\yy), \mc{L}_\xx (\zz/\yy)\r_{\si_\beta}  = 2\,.
\end{equation}
Similarly, we derive  $\mc{L}_{\zz}\left(I\otimes \mathcal{A}^{\rm full}_{\rm b}\right)=0$, and
\begin{equation} \label{eq:lzeig}
\mc{L}_\zz \left(\xx/\yy\otimes \mathcal{A}^{\rm full}_{\rm b}\right) = - 2 \left(\xx/\yy\right)\otimes \mathcal{A}^{\rm full}_{\rm b}\,,\q  - \l (\xx/\yy), \mc{L}_\zz (\xx/\yy)\r_{\si_\beta}  = 2\,.
\end{equation}
Further, for the action of local Lindbladian $\mc{L}_{\sigma^x_j/\sigma^y_j/\sigma^z_j}$ on the decomposition \eqref{eq:subspace_decom}, we have the following lemma, in analog with \cref{lem:L_gapped_property}.

\begin{lemma} \label{lem2}
The Lindbladian with Pauli coupling $\mc{L}_{\sigma^x_j}$, $\mc{L}_{\sigma^y_j}$, and $\mc{L}_{\sigma^z_j}$ ($j \ge 1$) are block diagonal with respect to the decomposition \eqref{eq:subspace_decom}:
\begin{align*}
    \mc{L}_{\sigma^x_j/\sigma^y_j/\sigma^z_j} \left(I/\xx/\yy/\zz \otimes \mc{A}_{\rm b}^{\rm full}\right) \subset I/\xx/\yy/\zz \otimes \mc{A}_{\rm b}^{\rm full}\,.
\end{align*}
In particular, it holds that for $j \ge 2$, $\mc{L}_{\sigma^x_j}(\mc{Q} \otimes \mc{A}_{\rm b}^{\rm full}) = \mc{Q} \otimes  \mc{L}_{\sigma^x_j}(\mc{A}_{\rm b}^{\rm full})$, that is,
\begin{equation*}
     \mc{L}_{\sigma^x_j}(Q A) = Q  \mc{L}_{\sigma^x_j}(A)\q \text{for $Q \in \mc{Q}$ and $A \in \mc{A}^{\rm full}_{\rm b}$}\,.
\end{equation*}
\end{lemma}

We now define the local Lindbladian:
\begin{equation} \label{eq:locallind}
    \wt{\mc{L}} = \sum_{j = 2}^N \mc{L}_{\sigma^x_j}
\end{equation}
which is primitive when restricted on $\mc{B}(\mc{H}_+) = \mc{A}^{\rm full}_{\rm b}$ \cite[Lemma 6]{Alicki_2009}. Thanks to the properties \eqref{eq:lxeig}--\eqref{eq:lzeig} of $\mc{L}_{\xx}$ and $\mc{L}_{\zz}$, the Lindbladian $\wt{\mc{L}}+\mc{L}_{\xx}+\mc{L}_{\zz}$ (as a part of $\mc{L}_\beta$ in \eqref{sampler:ising}) is primitive on $\mc{B}(\mc{H})$. Then, \cref{lem1} (item 2) readily gives
\begin{equation}\label{eqn:first_gap_lower_bound}
  \gap\left(\mc{L}_\beta\right)\geq \gap\left(\wt{\mc{L}}+\mc{L}_{\xx}+\mc{L}_{\zz}\right)\,.
\end{equation}
Thus, it suffices to consider the spectral gap of the latter one.
For this, we note from \eqref{eq:lxeig}--\eqref{eq:lzeig} that $\mathcal{L}_{\xx/\zz}$ is also block diagonal for the decomposition \eqref{eq:subspace_decom}. Then, by \cref{lem2}, we only need to estimate the gap of $\wt{\mc{L}}+\mc{L}_{\xx}+\mc{L}_{\zz}$ on each invariant subspace $I/\xx/\yy/\zz \otimes \mc{A}_{\rm b}^{\rm full}$:


\begin{itemize}
\item On $I \otimes \mc{A}^{\rm full}_{\rm b}$. Noting $\mc{L}_{\xx}|_{I \otimes \mc{A}^{\rm full}_{\rm b}}=\mc{L}_{\zz}|_{I \otimes \mc{A}^{\rm full}_{\rm b}}=0$, we have
\begin{equation*}
  \gap\left(\left(\w{\mc{L}} + \mc{L}_\xx+\mc{L}_\zz\right)\Big|_{I \otimes \mc{A}^{\rm full}_{\rm b}}\right) =  \gap\left(\w{\mc{L}}|_{I \otimes \mc{A}^{\rm full}_{\rm b}}\right)\,.
\end{equation*}

\item On $\xx/\yy/\zz \otimes \mc{A}^{\rm full}_{\rm b}$. Noting that any Davies generator is negative, we have
\begin{equation*}
  -\left(\w{\mc{L}} + \mc{L}_\xx+\mc{L}_\zz\right)\Big|_{\zz \otimes \mc{A}^{\rm full}_{\rm b}} \succeq -\left(\w{\mc{L}} + \mc{L}_\xx\right)\Big|_{\zz \otimes \mc{A}^{\rm full}_{\rm b}}\,.
\end{equation*}
Moreover, $\ker\big(\w{\mc{L}}|_{\zz \otimes \mc{A}^{\rm full}_{\rm b}}\big) = {\rm Span}(\zz \otimes I)$. It can be seen as follows. For any $\zz A\in \zz \otimes \mc{A}^{\rm full}_{\rm b}$ in the kernel, by the second part of \cref{lem2}, $\w{\mc{L}}(\zz A) = \zz \w{\mc{L}}(A)=0$. Thus, $A\in \ker\big(\w{\mc{L}}|_{\mc{A}^{\rm full}_{\rm b}}\big)$. Since $\w{\mc{L}}|_{\mc{A}^{\rm full}_{\rm b}}$ is primitive, it must be the case that $A=I$.
Then, by \cref{lem1} (item 4), there holds
 \begin{equation*}
        - \left(\w{\mc{L}} + \mc{L}_{\xx}\right)\Big|_{\zz \otimes \mc{A}^{\rm full}_{\rm b}} \succeq \frac{- \gap\left(\w{\mc{L}}|_{I \otimes \mc{A}^{\rm full}_{\rm b}}\right) \left\l \zz , \mc{L}_{\xx}(\zz ) \right\r_{\si_\beta}}{\gap\left(\w{\mc{L}}|_{I \otimes \mc{A}^{\rm full}_{\rm b}}\right) + \norm{\mc{L}_{\xx}}} = \Theta\left(\frac{\gap\left(\w{\mc{L}}|_{I \otimes \mc{A}^{\rm full}_{\rm b}}\right)}{\gap\left(\w{\mc{L}}|_{I \otimes \mc{A}^{\rm full}_{\rm b}}\right) + 1}\right)\,.
\end{equation*}
The same estimates hold for $- \left(\w{\mc{L}} + \mc{L}_{\xx}\right)\Big|_{\yy \otimes \mc{A}^{\rm full}_{\rm b}}$ and $- \left(\w{\mc{L}} + \mc{L}_{\zz}\right)\Big|_{\xx \otimes \mc{A}^{\rm full}_{\rm b}}$.
\end{itemize}
Combining the above arguments with \eqref{eqn:first_gap_lower_bound}, we find
\begin{align*}
     \gap\left(\mc{L}_\beta\right) &\geq \min\left\{\gap\left(-\left(\w{\mc{L}} + \mc{L}_\xx+\mc{L}_\zz\right)\middle|_{I \otimes \mc{A}^{\rm full}_{\rm b}}\right),\lambda_{\min}\left(-\left(\w{\mc{L}} + \mc{L}_\xx+\mc{L}_\zz\right)|_{\xx/\yy/\zz \otimes \mc{A}^{\rm full}_{\rm b}}\right)\right\} \\ & = \Theta\left(\gap\left(\w{\mc{L}}|_{I \otimes \mc{A}^{\rm full}_{\rm b}}\right)\right)\,,
\end{align*}
as $\gap\left(\w{\mc{L}}|_{I \otimes \mc{A}^{\rm full}_{\rm b}}\right) \to 0$. It has been proved in \cite[Proposition 1]{Alicki_2009} that for any $N$,
\begin{equation*}
    \gap\left(\w{\mc{L}}|_{I \otimes \mc{A}^{\rm full}_{\rm b}}\right) \ge \frac{e^{- 4 \beta J}}{e^{-4 \beta J} + 1}\,.
\end{equation*}
Therefore, to prove \cref{thm:ising_gap}, it suffices to focus on the gap of $\wt{\mc{L}}$ on $\mc{A}^{\rm full}_{\rm b}$ and show
\begin{equation}\label{eqn:gap_wt_L}
  \gap\left(\w{\mc{L}}|_{I \otimes \mc{A}^{\rm full}_{\rm b}}\right) = \mmg{\Omega\left(N^{-2}\right)}\,,  \quad \text{when}\ \beta \ge \Omega(\log N)\,,
\end{equation}
which can be done as \cref{lem:gap_abelian}. We have completed the proof of \cref{thm:ising_gap}.

\end{document}